\documentclass[onecolumn,draftclsnofoot]{IEEEtran}

\usepackage{geometry}
\usepackage{amsmath, amsthm, amsfonts, amssymb, multirow}
\theoremstyle{theorem}
\newtheorem{thm}{Theorem}[]
\newtheorem{coro}{Corollary}[]
\newtheorem{lem}{Lemma}

\theoremstyle{remark}
\newtheorem{rem}{Remark}
\theoremstyle{definition}
\newtheorem{defn}{Definition}

\global\long\def\P#1{\operatorname{P}\left\{  #1\right\}  }

\newcommand{\eps}{\varepsilon}

\def\setR{\mathbbm{R}}

\def\P{\mathcal{P}}
\def\E{\mathcal{E}}
\def\D{\mathcal{D}}
\def\C{\mathcal{C}}
\def\N{\mathcal{N}}

\newcommand{\etc}{\emph{i.e.}, }
\newcommand{\eg}{\emph{e.g.}, }

\newcommand{\argmin}[1]{\underset{#1}{\operatorname{argmin}}}

\newcommand{\Eox}[1]{{\rm E}\left[#1\right]}

\newcommand{\Lp}[3]{{\left\|{#1}\right\|}_{#2}^{#3}}

\newcommand{\LKr}[1]{\left\{#1\right\}}
\newcommand{\LPr}[1]{\left(#1\right)}
\newcommand{\LCr}[1]{\left[#1\right]}
\newcommand{\LPd}[1]{\left<#1\right>}
\newcommand{\Labs}[1]{\left|#1\right|}

\newcommand{\mvec}[1]{{\boldsymbol #1}}
\newcommand{\mmat}[1]{{\pmb{\rm #1}}}

\newcommand{\gos}{\rightarrow}
\newcommand{\baks}{\leftarrow}

\usepackage{graphicx}
\usepackage{color,amsmath,setspace}
\usepackage{amssymb,amsfonts,amscd,dsfont,bbm,color}
\usepackage{algorithm}
\usepackage[noend]{algpseudocode}

\newtheorem{assumption}{Assumption}

\definecolor{amethyst}{rgb}{0.6, 0.4, 0.8}

 \title{An efficient algorithm for compression-based compressed sensing}
 \author{Sajjad Beygi, Shirin Jalali, Arian Maleki, Urbashi Mitra} 
\begin{document}
\maketitle

\begin{abstract}
Modern image and video compression codes employ elaborate structures existing in such signals  to encode them into few number of bits. Compressed sensing recovery algorithms on the other hand use  such signals' structures to recover them from few linear observations. Despite the steady progress in the field of compressed sensing,  structures that are often used for signal recovery are still much simpler than those employed by  state-of-the-art compression codes. The main goal of this paper is to bridge this gap through answering the following question:  Can one employ a given compression code to build  an efficient (polynomial time) compressed sensing recovery algorithm? 
In response to this question, the compression-based gradient descent (C-GD) algorithm is proposed. C-GD, which is  a low-complexity iterative algorithm,   is able to employ a generic compression code for compressed sensing and therefore elevates  the scope of  structures used in compressed sensing to those used by compression codes. The convergence performance of C-GD and its required number of measurements in terms of the rate-distortion performance of  the compression code are theoretically analyzed. It is also shown that C-GD is robust to additive white Gaussian noise. Finally, the presented simulation results show that combining C-GD with commercial image compression codes such as JPEG2000 yields state-of-the-art performance in imaging applications. 
\end{abstract}

\section{Introduction}
The main problem of compressed sensing is to recover an unknown target signal $\mvec{x}\in\setR^n$ from undersampled linear measurements $\mvec{y}\in\setR^m$, 
\begin{align}
\label{linear_system_EQ}
\mvec{y}= \mmat{A}\mvec{x} + \mvec{z},
\end{align}
where    $\mmat{A}\in\setR^{m\times n}$ and  $\mvec{z}\in\setR^m$ denote   the measurement matrix and the measurement  noise, respectively. Since $m<n$, $\mvec{x}$ cannot be recovered accurately, unless we leverage some prior information about the signal $\mvec{x}$. Such information can mathematically be expressed by assuming that $\mvec{x} \in {\cal S}$, where ${\cal S} \subset \setR^{n}$ is a known set. 
 Intuitively, it is expected that, the ``smaller'' the set ${\cal S}$, the fewer number of measurements $m$ required for recovering $\mvec{x}$. In other words,  having more information about the target signal $\mvec{x}$ limits the feasible set in a way that enables the recovery algorithm  to identify a high-quality estimate of $\mvec{x}$ using fewer measurements.  In the last decade, researchers have explored several different instances of  set ${\cal S}$, such as the class of sparse signals or low-rank matrices \cite{Donoho:06,CandesR:06,CandesT:06, beygi2015multi, beygi2015nested}. Despite the mathematical elegance of these research directions, their impact has not yet met expectations in many application areas. A major barrier is that real-world signals often exhibit far more complex  structures than the simple ones that are studied in theoretical researches. In response to this shortcoming, we propose a new complimentary approach to elevate the scope of compressed sensing much beyond sparsity, low-rankness, etc. Our approach hinges upon the following simple assumption:  for a given  class of signals, there exists an \emph{efficient}  compression algorithm  that is able to represent the signals in that class with few number of  bits per symbol. In many application areas, such as  image and video processing, thanks to the extensive research performed in the past fifty years, such compression algorithms exist and  use sophisticated structures shared by signals in a class. Compressive sensing recovery algorithms that take advantage of similar complex structures potentially considerably  outperform current algorithms. This raises the following question that we address in this paper: How can one effectively employ a compression algorithm  in a compressed sensing recovery method? 
 
 In response to this question, we propose an iterative compression-based compressed sensing  recovery algorithm, called compression-based gradient descent (C-GD). C-GD, with no extra effort, elevates the scope of the structures used in signal recovery from simple structures, such as sparsity, to the more sophisticated ones that are used in the state-of-the-art compression codes. Our simulation results confirm that combining C-GD with the state-of-the-art image compression codes such as JPEG2000 yields state-of-the-art performance in compressive imaging applications. In addition to its state-of-the-art performance in the imaging applications, C-GD comes with a theoretical framework that derives its convergence rate, analyzes its performance in the presence of noise and other non-idealities in the system, and determines the impact of different parameters of C-GD on its performance. 

Using a compression code for the compressed sensing problem was first discussed in \cite{jalali2016compression} and \cite{rezagah2016using}. Both papers studied this problem   from a theoretical standpoint, for deterministic signal models and for stationary processes, respectively. In  \cite{jalali2016compression},  inspired by the Occam's razor principle, the compressible signal pursuit (CSP) optimization was proposed. We will review this algorithm briefly in Section \ref{ssec:csp}. The goal of the CSP optimization as a compression-based recovery algorithm is to, among the compressible signals, find the one that gives the lowest measurement error.  Theoretically, it can be  proved that, asymptotically, as $n$ goes to infinity and the  per-symbol reconstruction distortion approaches zero,   the CSP optimization achieves    the optimal sampling rate, in cases where such optimal rates are already known in the literature. While these results theoretically support  the idea of designing compression-based compressed sensing algorithms, they do not provide a recipe of how to achieve this goal in practice.    Solving the CSP  requires an exhaustive search over all compressible signals, or in other words, all signals in the codebook, whose size is often too large. For instance if we consider the class of images that can be compressed in $1000$ bits only, then the size of the codebook will be $2^{1000}$. Hence, CSP cannot be used in most real-world applications.  On the other hand,  the C-GD algorithm that is proposed in this paper is inspired by   projected gradient descent. Each step of the algorithm involves two operations: i) moving in the direction of the gradient of $\|\mvec{y} -\mmat{A}\mvec{x}\|^2$, and ii) projecting on the set of compressible signals. Both steps can be performed efficiently.

There are several other works in the literature that aim at elevating the scope of signal recovery beyond sparsity and low-rankness, and hence, in spirit, have a similar objective as this paper \cite{ChRePaWi10, metzler2014denoising, RichModelBasedCS}. Here we briefly mention those results and describe how they compare with our work. To use the framework developed in \cite{ChRePaWi10} for a class of signals  ${\cal S}$, one should construct a set of atoms such that  every $\mvec{x} \in {\cal S}$ can be represented using a few of those atoms. Finding non-standard atoms for real-world signals, such as images, is a sophisticated task and hence this approach has not found any application in compressive imaging or video recoding yet. The work of \cite{metzler2014denoising} is only concerned with independently and identically distributed (i.i.d.) measurement matrices. While in our theoretical results, we also only consider i.i.d.~measurement matrices, as  shown later  in the section on simulation results, C-GD works well with partial Fourier matrices  that are used for MRI and radar applications too.  On the contrary,  D-AMP developed in \cite{metzler2014denoising} fails to work with such matrices. The framework developed in \cite{RichModelBasedCS} still aims to use sparsity. However, it assumes that many of ${n \choose k}$ subspaces of $k$-sparse vectors, do not belong to our signal space. This framework is still only capable of using simple structures and hence does not offer the state-of-the-art performance in applications of compressed sensing, such as imaging. This point is demonstrated in \cite{metzler2014denoising}. 

Finally, there are many papers in the literature that have considered the compressive imaging application and proposed heuristic algorithms to employ sophisticated structures in the signal. (See \cite{dong2014compressive} and the references therein for a complete list of references.) In the simulation results section, we compare the performance of our algorithm with the state-of-the-art heuristic algorithm NLR-CS proposed in \cite{dong2014compressive}.  It turns out that C-GD offers the state-of-the-art performance. We should also emphasize that our framework has several other advantages over this line of work: (i) It is general and can be applied to different applications with no extra effort, (ii) it comes with a theoretical framework that shows its robustness to noise and other non-idealities in the measurement process.   

The organization of the paper is as follows. Section \ref{sec:background} reviews some background information. Section \ref{sec:C-GD} presents the main contributions of the paper, namely the C-GD algorithms and its theoretical analysis. Section \ref{sec:standardsignals} consider the application of the C-GD algorithm to some standard classes of signals. Section \ref{sec:simulations} presents our simulation results, which show that the C-GD algorithm achieves state-of-the-art performance in image compressed sensing. Section \ref{sec:proofs} provides the proofs of the mains results of the paper.  Finally, Section \ref{sec:conclusion} concludes the paper. 



\section{Background}\label{sec:background}

In this section, we first review the notation used throughout the paper and also the definitions of Gaussian and sub-Gaussian random variables. Then we review the rate-distortion function of a compression code, and formally define the objective of compression-based signal recovery. Finally we review the CSP optimization, which is the first  proposed method for employing compression algorithms for compressed sensing.

\subsection{Notations and defenitions}

Scalar values are denotes by lower-case letters such as $x$. A column vector is denotes by bold letters as $\mvec{x}$. The $i$-th element of $\mvec{x}$ is denoted  by ${x}_i$. Given vector $\mvec{x}\in\mathds{R}^n$, $\|\mvec{x}\|_{\infty}=\max_{i} |x_i|$ and $\|\mvec{x}\|_{2}=\sqrt{\sum_{i=1}^nx_i^2}$ denote the $\ell_{\infty}$ norm and the $\ell_2$ norm of $\mvec{x}$, respectively.  A matrix is denoted by bold capital letters such $\mmat{X}$ and its $(i, j)$-th element by ${X}_{i,j}$. The transpose of $\mmat{X}$ is given by $\mmat{X}^{T}$.   $\sigma_{\max}(\mmat{A})$ denotes the maximum singular value of $\mmat{A}$. Calligraphic letters such as $\D$ and ${\cal C}$ denote sets. The size of a set ${\cal C}$ is denotes as $|{\cal C}|$. The unit sphere in $\mathds{R}^n$ is  denoted by $S^{n-1}$, i.e., $S^{n-1}\triangleq
 \{ \mvec{u} \in \mathbb{R}^n \ | \  \|\mvec{u}\|_2 =1 \}$. Throughout the paper $\log$ and $\ln $ refer to logarithm in base $2$ and natural logarithm, respectively.  Finally, we use $O$ and $\Omega$ notation, defined below, to describe the limiting behavior of certain quantities. $f(n) = O(g(n))$ as $n \rightarrow \infty$, if and only if there exist $n_0$ and $c$ such that for any $n> n_0$, $|f(n)| \leq c |g(n)|$. Likewise, $f(n) = \Omega(g(n))$ as $n \rightarrow \infty$, if and only if there exist $n_0$ and $c$ such that for any $n> n_0$, $|f(n)| \geq c |g(n)|$. 

In this paper, the main results are proved for both Gaussian and sub-Gaussian measurement matrices. In the following, we briefly review the definition  of  sub-Gaussian and sub-exponential  random variables. 
\begin{defn}[Sub-Gaussian] 
 A random variable $X$ is sub-Gaussian when
\[
\Lp{X}{\psi_2}{} \triangleq  \inf \LKr{L>0 : \Eox{\exp\LPr{\frac{|X|^2}{L^2}}} \le 2} < \infty.
\]
\end{defn}
Note that Gaussian random variables are also sub-Gaussian random variables. For $X\sim\N(0,\sigma^2)$, $\Lp{X}{\psi_2}{}{}=\sqrt{8\over 3}\sigma$.
\begin{defn}[Sub-exponential] 
\label{sub_expo_def}
A random variable $X$ is a sub-exponential random variable when 
\begin{align}
\Lp{X}{\psi_1}{} \triangleq \inf \LKr{L>0 : \Eox{{\rm e}^{\frac{|X|}{L}}} \le 2} < \infty.
\end{align}
\end{defn}
Using the above definitions, it is straightforward to show the following result. 
\begin{lem} \label{prod_subgaus}
Let $X$ and $Y$ be sub-Gaussian random variables. Then $XY$ is sub-exponential. Moreover, $ \Lp{XY}{\psi_1}{} \le \Lp{X}{\psi_2}{}  \Lp{Y}{\psi_2}{}. $
\end{lem}

\subsection{Objective of compression-based compressed sensing}

Consider  a  compact set ${\cal Q}\subset\setR^n$. (Throughout the paper, we focus on this deterministic signal model. However,  all the results can be extended to the stochastic setting as well. Refer to  \cite{rezagah2016compression} for further information.)
 A rate-$r$ lossy compression code for set ${\cal Q}$ is characterized   by its encoding and decoding mappings $(f,g)$, where
$$
f:{\cal Q} \gos \LKr{1,2,\ldots, 2^{r}}, 
$$ 
and
$$
g:\LKr{1,2,...,2^{r}}\gos {\cal Q}.
$$ 
 Encoding and decoding  mappings  $(f,g)$ define a codebook $\C$, where
\begin{align*}
\C&= \LKr{g\LPr{f\LPr{\mvec{x}}}: \; \mvec{x}\in{\cal Q}}.
\end{align*}
Clearly, $|\C|\leq 2^{r}$. The performance of  a code defined by $(f,g)$ is characterized by its i) rate $r$ and ii) maximum distortion $\delta$ defined as 
\[
\delta=\sup_{\mvec{x}\in{\cal Q}} \Lp{\mvec{x}-g(f(\mvec{x}))}{2}{}.
\]

The problem of compression-based compressed sensing can be formally stated in the following way. Suppose that a family of compression codes $(f_r, g_r)$, parameterized with rate $r$ is given for ${\cal Q}$. For instance, JPEG or JPEG2000 compression algorithms at different rates can be considered as a family of compression algorithms for the class of natural images. The deterministic distortion-rate function of this family is given by
\[
\delta (r)=  \sup_{\mvec{x}\in{\cal Q}} \Lp{\mvec{x}-g_r(f_r(\mvec{x}))}{2}{}.
\]
Note that for any reasonable code,  $\delta(r)$ is expected to be a monotonically non-increasing function of $r$. We also define the deterministic rate-distortion function of the compression code as
\[
r(\delta) = \inf \{r : \delta(r) \leq \delta \}. 
\]

\begin{rem}
Typically  a family of  compression codes is defined as a sequence of compression codes that are indexed by blocklength $n$ and operate either at a fixed rate or a fixed distortion. In this paper, we are more interested in the setting where the blocklength is fixed and the rate or the distortion changes. Therefore, we consider a family of compression codes with fixed blocklength $n$ indexed by rate $r$. 
\end{rem}

Based on these definitions we can formally state the objectives of this paper as follows. Consider the problem of compressed sensing of signals in a set ${\cal Q}$. Suppose that instead of knowing the set ${\cal Q}$ explicitly, for signals in ${\cal Q}$, we have access to a family of compression algorithms with rate-distortion function $r(\delta)$. For $\mvec{x}\in{\cal Q}$ our goal is not to compress it, but to recover it from its undersampled set of linear projections $\mvec{y} = \mmat{A} \mvec{x}+\mvec{z}$. The goal of compression-based compressed sensing is summarized in the following two questions:

\textbf{Question 1:} {\it Can one employ a given compression code in an efficient (polynomial time) signal recovery algorithm?} 

\textbf{Question 2:} {\it Can we characterize the number of observations such an algorithm requires to accurately recover $\mvec{x}$, in terms of the rate-distortion performance of the code?}

Note that the algorithms that are developed in response to the above questions will automatically employ the structure captured by the compression algorithms, and hence can immediately elevate the scope of compressed sensing much beyond simple structures. For instance, an MPEG4-based compressed sensing recovery algorithm will use not only the intra-frame but also the inter-frame dependencies among different pixels. 
 \subsection{Compressible signal pursuit}\label{ssec:csp}
 
Consider the following simplified version of our main questions: Can one employ a compression code for signal recovery? Can we characterize the number of observations such algorithm requires for an accurate recovery of $\mvec{x}$ in terms of the rate-distortion performance of the code? Note that the only simplification is that we have removed the constraint on the computational complexity of the recovery scheme. In response to this simplified question,   \cite{jalali2016compression} proposed the compressible signal pursuit (CSP) optimization that estimates signal $\mvec{x}$ based on  measurements $\mvec{y}$ as follows.  Among all the signals $\mvec{u} \in {\cal Q}$ that satisfy the measurement constraint, i.e.~ $\mvec{y} = \mmat{A}\mvec{u}$,  it searches for the one that can be compressed well by the compression code described by  $(f_r,g_r)$. More formally, given a lossy compression code with codebook $\C_r = \{g_r (f_r(\mvec{x}) : \mvec{x} \in {\cal Q} \}$, the CSP optimization recovers $\mvec{x}$ from its measurements $\mvec{y} = \mmat{A}\mvec{x}$ as follows
\begin{align}
	\label{search_optimization}
	\hat{\mvec{x}} = \argmin{\mvec{u}\in \C_r}\Lp{\mvec{y} - \mmat{A}\mvec{u}}{2}{2}. 
\end{align}
The performance of the CSP optimization is characterized in \cite{jalali2016compression} and \cite{rezagah2016using}, under deterministic and stochastic signal models, respectively. Before we mention the theoretical results, we should emphasize that at this point CSP is based on an exhaustive search over the codebook and hence is computationally infeasible. The following result  from \cite{jalali2016compression} characterizes the   performance of the CSP optimization in the noiseless setting, where $\mvec{z}=\mvec{0}$.
\begin{thm}[Corollary 1 in \cite{jalali2016compression}]\label{thm:CSP}
Consider a family of compression codes $(f_r,g_r)$ for set ${\cal Q}$ with corresponding codebook ${\cal C}_r$ and rate-distortion function $r(\delta)$. Let $\mmat{A}\in\setR^{m\times n}$, where $A_{i,j}$ are i.i.d.~${\cal N}(0,1)$. For $\mvec{x}\in{\cal Q}$ and $\mvec{y}=\mmat{A}\mvec{x}$, let $\hat{\mvec{x}}$  denote the solution of \eqref{search_optimization}. Given $\nu>0$ and $\zeta>1$, such that ${\zeta \over \log{1\over {\rm e}\delta}}<\nu$, let 
\[
m= {\zeta r\over \log {1\over  {\rm e}\delta }}.
\] 
Then,
\[
P( \Lp{\mvec{x}- \hat{\mvec{x}}}{2}{} \geq \theta\delta^{1-{1+\nu \over \zeta}})\leq {\rm e}^{-0.8 m}+{\rm e }^{-0.3 \nu r},
\]
where $\theta=2{\rm e}^{-(1+\nu)/\zeta}$.
\end{thm}
For a simpler interpretation of this result define the $\alpha$-\emph{dimension} of a family of compression codes with rate-distortion function $r(\delta)$ as
\[
\alpha = \lim \sup_{\delta \rightarrow 0} \frac{r(\delta)}{\log (1 /\delta)}.
\]
Suppose that a small value of $\delta$ (or large value of $r$) is used in CSP. Then, roughly speaking, Theorem \ref{thm:CSP} implies that CSP returns an almost accurate estimate of $\mvec{x}$ as long as $m> \alpha $. Note that $\alpha$  is usually much smaller than $n$, and hence the number of measurements CSP requires is much smaller than the ambient dimension of the signal. (See Section \ref{sec:standardsignals} for some classical examples.)

\begin{rem}
In this paper, we focus entirely on the deterministic setting. The performance of CSP in the stochastic setting is studied in \cite{rezagah2016compression}. Using the connection between the  R\'enyi information dimension and rate distortion dimension of a random variable  \cite{kawabata1994rate},  it has been proved in \cite{rezagah2016compression, jalali2016new} that for i.i.d.~sources with a mixture of discrete and continuous distribution, CSP achieves the optimal sampling rate. 
\end{rem}

\begin{rem}
The robustness of CSP to deterministic and stochastic measurement noises has also been proved in  \cite{jalali2016compression}. However, for the sake of brevity we do not repeat those results here and only discussed the noiseless setting in Theorem \ref{thm:CSP}.
\end{rem}

Unfortunately,  the positive  theoretical properties of CSP are overshadowed by the fact that  implementing it requires an exhaustive search over the set of all codewords. This makes CSP computationally infeasible for  all real-world applications. In the next section, we propose an efficient CS recovery algorithm that employs compression code and compare its performance with that of CSP. 

\section{Our main contributions}\label{sec:C-GD}

\subsection{Compression-based gradient descent (C-GD)}

As discussed in the last section, CSP is based on an exhaustive search and is computationally infeasible for real-world signals. In response to this drawback of CSP, we propose a computationally efficient and theoretically analyzable approach to approximate the solution of CSP. Towards this goal, inspired by the \emph{projected gradient descent} (PGD) algorithm \cite{rockafellar1976monotone}, we propose the following iterative algorithm: Start from some $\mvec{x}^{0}\in\mathds{R}^n$. For $k=1,2,\ldots$,
\begin{align}
	\label{update_rule_pgdMethod}
	\mvec{x}^{k+1} \leftarrow\, & \P_{\C_r}\LPr{\mvec{x}^{k} + \eta_k \,\mmat{A}^{T}\LPr{\mvec{y} - \mmat{A}\mvec{x}^k}},
\end{align}
where
\begin{align}
\P_{\C_r}\LPr{\mvec{z}} = \argmin{\mvec{u}\in \C_r}\Lp{\mvec{u} - \mvec{z}}{2}{2}.
\end{align}
Here index $k$ denotes the iteration number and $\eta_k \in \setR$ denotes the step size. 
We refer to this algorithm as \emph{compression-based gradient descent} (C-GD). Each iteration of this  algorithm involves performing  two operations. In the first step, it moves in the direction of the negative of $\|\mvec{y}-\mmat{A}\mvec{x}\|_2^2$ to find solutions that are closer to the  $\mvec{y}=\mmat{A}\mvec{u}$ hyperplane. The second step, i.e.,~the projection step, ensures that the estimate C-GD obtains belongs to the codebook.

The first step of C-GD is straightforward and requires two matrix-vector multiplications. For the second step, ideally,  applying the encoder and decoder to a signal $\mvec{x}$ yields   the closest codeword of the compression code. Hence, we make the following assumption about the compression code to $\mvec{x}$.
  \begin{assumption}
In analyzing the performance of C-GD, we assume that the compression code $(f_r,g_r)$ satisfies
\begin{align}\label{approximatedProjection}
\P_{\C_r}\LPr{\mvec{x}}=g_r(f_r(\mvec{x})).
\end{align}
\end{assumption}
 Under Assumption 1, the projection step of C-GD can be implemented efficiently. More precisely, under this assumption, the C-GD algorithm is simplified to 
\begin{align}
	\label{update_rule_pgdMethod}
	\mvec{x}^{k+1} \leftarrow\, &  g_r\left(f_r \LPr{\mvec{x}^{k} + \eta_k \,\mmat{A}^{T}\LPr{\mvec{y} - \mmat{A}\mvec{x}^k}}\right).
\end{align}
Hence, each step of this algorithm requires two matrix-vector multiplication and an application of the encoder and the decoder of the given compression code. In the next section, we summarize our theoretical results regarding the performance of C-GD. Note that for the notational simplicity, we present all our results under Assumption 1. However as will be discussed after Corollary \ref{Corollary_Gauss}, we can relax this assumption and still analyze the iterative algorithm proposed in \eqref{update_rule_pgdMethod}. 


\subsection{Convergence Analysis of C-GD }
\label{sec:converge_analysis}

The objective of this section is to theoretically analyze some of the properties of C-GD. As discussed before, the measurement vector is denoted with $\mvec{y}= \mmat{A} \mvec{x}+ \mvec{z}$, where  $\mvec{x} \in{\cal Q}$, $\mmat{A} \in \setR^{m\times n}$, and $\mvec{z}$ is the noise. Furthermore, we assume that a family of compressions codes $(f_r, g_r)$ parameterized with the rate $r$ that is known for ${\cal Q}$. Starting with $\mvec{x}^{0}$, C-GD uses iterations
\[
\mvec{x}^{k+1} \leftarrow\,  \P_{\C_r}\LPr{\mvec{x}^{k} + \eta \,\mmat{A}^{T}\LPr{\mvec{y} - \mmat{A}\mvec{x}^k}}.
\] 
to obtain a good estimate of $\mvec{x}$. In our theoretical analysis of C-GD, we focus on popular measurement matrices in compressed sensing area namely dense i.i.d. Gaussian and sub-Gaussian matrices. In Section \ref{sec:simulations}, we present our simulation results that confirm the success of C-GD for Fourier matrices as well. However, the theoretical study of this important class of matrices is left for future research. 
We first study the performance of C-GD for i.i.d.~Gaussian measurement matrices. Then, we extend our results to i.i.d.~sub-Gaussian measurement matrices. 

\begin{thm}
\label{noisy_dense_thm_iid_gauss}
Let $\mmat{A} \in \setR^{m\times n}$ be a random Gaussian measurement matrix with i.i.d entries ${A}_{i,j} \sim \N\LPr{0, \sigma_a^2}$, and $\mvec{z} \in\setR^m$ be an i.i.d. Gaussian noise vector with $z_i \sim \N\LPr{\mvec{0}, \sigma_z^2}$. Let $\eta = \frac{1}{m\sigma_a^2}$ and define $ \tilde{\mvec{x}} =\P_{\C_r}(\mvec{x})$, where $\P_{\C_r}(\cdot)$ is defined in \eqref{approximatedProjection}. If Assumption 1 holds, then,  given $\epsilon>0$, for  $m\ge 80r\LPr{1+ \epsilon},$  with a probability larger than $1-2^{-2\epsilon r+1}$, we have
\begin{align}
\|\mvec{x}^{k+1} - \tilde{\mvec{x}}\|_2 \le  0.9\|\mvec{x}^k - \tilde{\mvec{x}}\|_2 +  2\LPr{2+\sqrt{\frac{n}{m}}}^2\delta + \frac{\sigma_z}{\sigma_a}\sqrt{\frac{8(1+\epsilon)r}{m}}, 
\end{align}
for $k=0,1,2,\ldots$. 
\end{thm}

The proof of Theorem \ref{noisy_dense_thm_iid_gauss} is given in Section \ref{proof_noisy_dense_thm_iid_gauss}.%

A few important features of this theorem are discussed in the following remarks.

\begin{rem}
According to Theorem \ref{noisy_dense_thm_iid_gauss}, C-GD requires $\Omega(r(\delta))$ measurements (for small values of $\delta$) to obtain an accurate estimate. Hence, according to this theorem, even in the noiseless setting we should not let $\delta \rightarrow 0$. Otherwise, for many signal classes $r(\delta) \rightarrow \infty$ and hence C-GD will require more measurements than the ambient dimension of $\mvec{x}_o$. As we will demonstrate in several examples in Section \ref{sec:standardsignals}, one can set $\delta$ to a small dimension-dependent value, e.g. $\delta = 1/n$, to ensure that C-GD can obtain a very good estimate with a few observations. Section \ref{sec:standardsignals} studies how $\delta$ is set and connects Theorem \ref{noisy_dense_thm_iid_gauss} to some classical results in compressed sensing.    
\end{rem}

\begin{rem}
According to Theorem \ref{thm:CSP}, CSP requires $\Omega\Big(\frac{r(\delta)}{\log(1/\delta)}\Big)$. This implies that, in the noiseless setting, for fixed $n$ and $m$, the estimate of CSP improves as $\delta$ decreases. However, this seems not to be the case for C-GD. According to Theorem \ref{noisy_dense_thm_iid_gauss}, C-GD requires more than $\Omega(r(\delta))$.  Hence, as $\delta$ decreases, C-GD requires more measurements. This is not an issue in almost all the applications of compressed sensing, where $n$ is very large and hence there is not much difference between setting $\delta = 1/n$ and $\delta=0$. However, from a theoretical perspective it is interesting to discover whether this mismatch is an artifact of our proof technique or it is a fundamental loss that is incurred by the reduction in the computational complexity. This question is left for future research. 
\end{rem}

We postpone the discussion of the convergence rate and the reconstruction error to Section \ref{sec:standardsignals}, where we discuss some classical examples and compare the conclusion of this theorem with some classical results in compressed sensing. In the setup considered in  Theorem \ref{noisy_dense_thm_iid_gauss}, as $n$ increases the per measurement's signal to noise ratio (SNR) is approaching infinity. Note that by scaling the measurement matrix, we can also obtain results for fixed SNR. The next corollary clarifies our claim.

\begin{coro}
\label{Corollary_Gauss}
Consider the  setup  of Theorem \ref{noisy_dense_thm_iid_gauss}, where now  $\mmat{A} \in \setR^{m\times n}$ is a random Gaussian measurement matrix with i.i.d entries ${A}_{i,j} \sim \N\LPr{0, \frac{\sigma_a^2}{n}}$ and $\eta={n \over \sigma_a^2 m}$. If Assumption 1 holds, then, given $\epsilon>0$, for  $m\ge 80r\LPr{1+ \epsilon}$, with a probability larger than $1-2^{-2\epsilon r+1}$, for $k=0,1,2,\ldots$, we have
\begin{align}\label{eq:noisyfromcor}
\frac{1}{\sqrt{n}}\|\mvec{x}^{k+1} - \tilde{\mvec{x}}\|_2 \le  \frac{0.9}{\sqrt{n}}\|\mvec{x}^k - \tilde{\mvec{x}}\|_2 +  2\LPr{2+\sqrt{\frac{n}{m}}}^2{\delta \over \sqrt{n}} + \frac{\sigma_z}{\sigma_a}\sqrt{\frac{8(1+\epsilon)r}{m}}.
\end{align}
\end{coro}

Finally, Assumption 1 seems to play a critical role in Theorem \ref{noisy_dense_thm_iid_gauss} and Corollary \ref{Corollary_Gauss}. However, thanks to the linear convergence of C-GD one can relax Assumption 1 in several ways and still obtain recovery guarantees for C-GD. For the sake of brevity we only mention one such result in the paper.

\begin{thm}\label{thm:imperfectproject}
Consider the  setup  of Theorem \ref{noisy_dense_thm_iid_gauss}, where now  $\mmat{A} \in \setR^{m\times n}$ is a random Gaussian measurement matrix with i.i.d entries ${A}_{i,j} \sim \N\LPr{0, \frac{\sigma_a^2}{n}}$ and $\eta={n \over \sigma_a^2 m}$. Suppose that $\sup_{\mvec{x}} \|g_r(f_r(\mvec{x})) - \mathcal{P}_{\mathcal{C}_r}(\mvec{x})\|_2 \leq \xi$. Then, given $\epsilon>0$, for  $m\ge 80r\LPr{1+ \epsilon}$, with a probability larger than $1-2^{-2\epsilon r+1}$, for $k=0,1,2,\ldots$, we have
\begin{align}\label{eq:noisyfromcor}
\frac{1}{\sqrt{n}}\|\mvec{x}^{k+1} - \tilde{\mvec{x}}\|_2 \le  \frac{0.9}{\sqrt{n}}\|\mvec{x}^k - \tilde{\mvec{x}}\|_2 +  2\LPr{2+\sqrt{\frac{n}{m}}}^2{\delta \over \sqrt{n}} + \frac{\sigma_z}{\sigma_a}\sqrt{\frac{8(1+\epsilon)r}{m}} + \frac{\xi}{\sqrt{n}}.
\end{align}
\end{thm}

The proof can be found in Section \ref{ssec:proof:thm:imperfectproject}. Note that at every iteration the imperfect projection introduces an error. These errors accumulate as the algorithm proceeds. However, thanks to the linear convergence of the algorithm the overall error caused by the imperfect projection remains at the order of $O(\xi/\sqrt{n})$. 

 All our results so far have been stated for Gaussian measurement matrices. However, they can be generalized to subgaussian matrices too. To prove our claim we extend one of our results, i.e. Theorem \ref{noisy_dense_thm_iid_gauss},  to sub-Gaussian matrices below. 
Our next theorem shows that Theorem \ref{noisy_dense_thm_iid_gauss} can be extended to sub-Gaussian measurement matrices. 
 
 \begin{thm}
\label{noisy_dense_thm_iid_sub}
Let $\mmat{A} \in \setR^{m\times n}$ be a zero-mean random sub-Gaussian measurement matrix with i.i.d entries, such that $\Lp{{A}_{i,j}}{\psi_2}{} \le K$ and $\Eox{{A}_{i,j}^2}=\sigma_a^2$. The noise vector $\mvec{z} $ is distributed as $\N\LPr{\mvec{0}, \sigma_z^2\mmat{I}_{m\times m}}$. Set $\eta = \frac{1}{m\sigma_a^2}$. Then, given $\epsilon>0$ and $\mu_0\in(0,1)$, such that $\mu_0 \sigma_a^2\leq 2K^2$, for 
$$
m >\LPr{\frac{16K^4(1+\epsilon)}{\mu_0^2  \sigma_a^4 \log {\rm e}   }}r,
$$ 
with probability at least 
\[
1- 2^{-4r\epsilon}-{\rm e}^{-{m\over 4}}- 2^{-2r\epsilon},
\]
for $k=0,1,2,\ldots$, 
\begin{align}
\|\mvec{x}^{k+1} - \tilde{\mvec{x}}\|_2 \le \mu_0\|\mvec{x}^k - \tilde{\mvec{x}}\|_2 +   8\LPr{1+\frac{3Kn}{\sigma_a^2m}}\delta  + {9K\sigma_z\over \sigma_a^2} \sqrt{  r(1+\epsilon)\over m }.
\end{align}
\end{thm}

Proof is given in Section \ref{proof_noisy_dense_thm_iid_sub}. Since all the terms in Theorem \ref{noisy_dense_thm_iid_sub} are similar to the corresponding terms in Theorem \ref{noisy_dense_thm_iid_gauss} we do not discuss them here.  We only discuss the convergence rate. Note that the convergence rate $\mu_0$ has a direct impact on the number of measurements. If we want the convergence to be fast ($\mu_0$ to be small), we should either increase the number of measurements or decrease the rate $r$. If we decrease the rate, the two error terms $ 8\LPr{1+\frac{3Kn}{\sigma_a^2m}}\delta  + {9K\sigma_z\over \sigma_a^2} \sqrt{  r(1+\epsilon) \over m }$ grow.

%

\section{Standard signal classes}\label{sec:standardsignals}
In this section we discuss the corollaries of our main theorem for the following two standard signal classes that have been studied extensively in the literature: (i) sparse signals, and (ii) piecewise polynomials. For each class, we first construct a simple compression algorithm that can be efficiently implemented in practice, and then explain the implications of C-GD and its analysis for that class. These examples enable us to shed light on different aspects of C-GD, such as (i) convergence rate, (ii) number of measurements, (iii) reconstruction error in noiseless setting, and (iv) reconstruction error in the presence of noise.
\subsection{Sparse signals}
 Let $\mathcal{B}_p^n(\rho) \triangleq \{\mvec{x} \in \mathds{R}^n: \ \|\mvec{x}\|_p \leq \rho \}$ represent a ball of radius $\rho$ in $\mathds{R}^n$. Also, let $\Gamma_k^n$ denote the set of all $k$-sparse signals in $\mathcal{B}_p^n(\rho) $, i.e.,
\begin{equation}\label{eq:ksparse}
\Gamma_k^n \triangleq \{ \mvec{x} \in \mathcal{B}_p^n(1)  \ : \ \|\mvec{x}\|_0 \leq k\}.
\end{equation}
In order to apply C-GD to this class of signals, we  first need to construct a family of compression codes for such sparse bounded signals. Consider a family of compression codes for set $\Gamma_k^n$ defined as follows. For $\mvec{x}\in \Gamma_k^n$, (i) encode the locations of  its at most $k$ non-zero entries ($\approx \log{n \choose k})$ bits)  and (ii) apply a uniform quantizer to the magnitudes of the  non-zero entries. (Using $b$ bits for the magnitude of each entry and one bit for its sign, this step spends $(b+1)k$ bits.)  Using this specific compression algorithm in the C-GD framework yields an algorithm which is very similar to iterative hard thresholding (IHT)  \cite{blumensath2009iterative}.  At every iteration, after moving in the opposite direction of the gradient, the standard IHT algorithm   keeps the  $k$ largest elements and sets the rest to zero. The C-GD algorithm on the other hand, while having the same first step, performs the projection step in a slightly different manner. For the projection onto codewords, similar to  IHT,  it first finds the $k$ largest entries. Then, for each entry $x_i$, it first limits it between $[-1,1]$ by computing $x_i\mvec{1}_{x_i\in(-1,1)}+\mvec{1}_{x_i\geq1}-\mvec{1}_{x_i\leq -1}$. Then, it  quantizes the result by $b+1$ bits. The following corollary enables us to compare our results with that of hard thresholding. 

Consider  $\mvec{x}\in \Gamma_k^n$ and let $\mvec{y}=\mmat{A}{\mvec{x}}+\mvec{z}$, where ${A}_{i,j} \stackrel{\rm i.i.d.}{\sim} \N\LPr{0, {\sigma_a^2}/{n}}$ and ${z}_{i} \stackrel{\rm i.i.d.}{\sim} \N\LPr{0, {\sigma_z^2}}$. Let $\tilde{\mvec{x}}$ denote the projection of $\mvec{x}$ onto the codebook of the above-described code. The following corollary of Theorem \ref{noisy_dense_thm_iid_gauss} characterizes the convergence performance of C-GD applied to $\mvec{y}$ when using  this code.

\begin{coro}\label{cor:k-spaese}
Given $\gamma>0$, set the quantization level of the compression code  to $b+1=\lceil \gamma \log n +{1\over 2}\log k\rceil+1$ bits. Also, set $\eta={n \over \sigma_a^2 m}$. Then, given $\epsilon>0$,   for $m\ge 80 \tilde{r} (1+ \epsilon)$, where $\tilde{r} = (1+\gamma) k \log n + {k\over2}\log k + 2k$, 
\begin{align}\label{eq:noisyfromcor_ksparse}
\frac{1}{\sqrt{n}}\|\mvec{x}^{t+1} - \tilde{\mvec{x}}\|_2 \le  \frac{0.9}{\sqrt{n}}\|\mvec{x}^t - \tilde{\mvec{x}}\|_2 +  2\LPr{2+\sqrt{\frac{n}{m}}}^2{n^{-1/2-\gamma}} + \frac{\sigma_z}{\sigma_a}\sqrt{\frac{8(1+\epsilon)\tilde{r}}{m}},
\end{align}
for $t=1,2,\ldots$,   with probability larger than $1- 2^{-2\epsilon \tilde{r}}$.
\end{coro}
\begin{proof}
Consider $u\in[-1,1]$. Quantizing $u$ by a uniform quantizer that uses $b+1$ bits yields $\hat{u}$, which satisfies $|u-\hat{u}|<2^{-b}$. Therefore, using $b+1$ bits to quantize each non-zero element of $\mvec{x}\in\Gamma_k^n$ yields a code which achieves distortion $\delta\leq 2^{-b}\sqrt{k}$. Hence, for $b+1=\lceil \gamma \log n +{1\over 2}\log k\rceil+1$, 
\[
\delta\leq n^{-\gamma}.
\]
On the other hand, the code rate $r$ can be upper-bounded as
\[
r \leq \sum_{i=0}^k\log {n \choose i} + k(b+1)\leq \log  n^{k+1} + k(b+1)=(k+1)\log n+k(b+1),
\]
where the last inequality holds for all $n$ large enough. The rest of the proof follows directly from inserting  these numbers in the statement of   Theorem \ref{noisy_dense_thm_iid_gauss}.
\end{proof}

This corollary enables us to provide further intuition on the performance of C-GD.  We start with the noiseless observations and for the moment we only study the required number of measurements and the reconstruction error in the absence of noise. The number of measurements required by the C-GD algorithm is $m= \Omega(k \log n)$. For the final reconstruction error, we can use \eqref{eq:noisyfromcor_ksparse} and obtain 
\[
\lim_{t \rightarrow \infty} \frac{1}{\sqrt{n}}\|\mvec{x}^{t+1} - \tilde{\mvec{x}}\|_2 = O\left(\LPr{2+\sqrt{\frac{n}{m}}}^2n^{-1/2-\gamma} \right). 
\]
 This implies that  the recovery error satisfies 
\[
\lim_{t \rightarrow \infty} \frac{1}{\sqrt{n}}\|\mvec{x}^{t+1} - \tilde{\mvec{x}}\|_2 = O\left( \frac{n^{\frac{1}{2} - \gamma}}{m} \right).
\]
Hence, if $\gamma> 0.5$, the error vanishes as the dimension grows. Regarding the number of measurements, there are two  points that we would like to emphasize here:
\begin{enumerate}
\item The $\alpha$-dimension of the above-described code  is $k$. Hence, CSP is able to accurately recover the signal from only $k$ measurements.  However, solving CSP  requires an exhaustive search over all the codebooks. On the other hand, C-GD requires $k\log n$ measurements. Therefore, it seems that  the extra  $\log n$ factor is the price for having  an efficient recovery algorithm. 

\item For large values of $n$, $n^{-\gamma}$ is very small, and hence C-GD becomes very similar to IHT. The results we have obtained for C-GD in this case are slightly weaker than those provided for IHT. First, our reconstruction is not exact even in the noiseless setting. Second, the number of measurements C-GD requires is $O(k\log n)$ compared to $O(k \log( n/k))$ required by IHT. These minor differences seem to be the price of the generality of the bounds derived for C-GD. 
\end{enumerate}

So far, we have studied two important quantities in Corollary \ref{Corollary_Gauss}, i.e., (i) required number of measurements, and (ii) reconstruction error in the absence of noise. The last important term, is the reconstruction error in the presence of the noise. From Corollary \ref{cor:k-spaese}, the distortion caused due to the presence of a Gaussian noise is $O\left(\frac{ \sigma_z}{\sigma_a}\sqrt{\frac{\tilde{r}}{m}}\right)$, or $ O\left(\frac{ \sigma_z}{\sigma_a}\sqrt{\frac{k \log n}{m}}\right)$. Note that there is no result on the performance of IHT in the presence of stochastic measurement noise. However, this noise sensitivity is comparable with the performance of algorithms that are based on convex optimization such as LASSO and Dantzig selector \cite{candes2007dantzig, bickel2009simultaneous}. 

\subsection{Piecewise polynomial functions}

Let ${\rm Poly}_N^Q$ denote the class of  piecewise-polynomial functions $p(\cdot):[0,1] \rightarrow [0,1]$  with at most $Q$  singularities\footnote{A singularity is a point at which the  function is not infinitely differentiable.}, where each polynomial has a maximum degree of $N$. For $p\in{\rm Poly}_N^Q$, let  $(x_{1}, x_{2}, \ldots,x_{n})$ be the samples of $p$ at 
\[
0, {1 \over n}, \ldots,{n-1\over n}.
\]  
For $\ell=1,\ldots,Q$, let $\{a_i^{\ell}\}_{i=0}^{N_{\ell}}$ denote the set of coefficients of the $\ell^{\rm th}$ polynomial in $p$. Here, $N_{\ell}\leq N $ denotes the degree of the $\ell$-th polynomial. For  notational simplicity,  assume that the coefficients of each polynomial belong to $[0,1]$ interval and that $\sum_{i=0}^{N_{\ell}} a^{\ell}_i <1$, for every $\ell$. Define 
\begin{align}
\mathcal{P} \triangleq \left\{\mvec{x} \in \mathds{R}^n \ | \ x_{i} = p(i/n), \ p \in {\rm Poly}_N^Q \right\}.\label{eq:P-def}
\end{align}
Note that this class of functions is a generalization of the class of piecewise-constant functions that are popular in many applications including imaging. To apply C-GD to this class of signals,  we need to design an efficient compression code for signals in $\mathcal{P}$ and describe how to project signals on the codewords of $\cal{P}$. For the first part, consider a simple code which for any signal $\mvec{x}\in\cal{P}$, it first   describes the locations of its discontinuities and then, using a uniform quantizer that spends $b$ bits per coefficient,  describes  the quantized   coefficients of the polynomials. For the other task, which is projecting a signal $\mvec{x}\in\mathds{R}^n$, Appendix  \ref{app:dynamicprogram} describes how  we can find the closest signal to $\mvec{x}$ in $\mathcal{P} $, i.e., \begin{eqnarray}\label{eq:dynamprog}
\tilde{\mvec{x}} = \arg\min_{\mvec{z} \in \mathcal{P}} \|\mvec{x}- \mvec{z}\|_2^2,
\end{eqnarray}
 using dynamic programing. Once that signal is found, its quantized version using the described code represents the desired projection. 

Note that C-GD combined with the described compression code is an extension of IHT to piecewise-polynomial  functions. At every iteration, C-GD projects its current estimate of the signal   to space of piecewise-polynomial functions. 

%
%
%
Consider $\mvec{x}\in\P$ and $\mvec{y}=\mmat{A}\mvec{x}+\mvec{z}$, where  ${A}_{i,j} \stackrel{\rm i.i.d.}{\sim} \N\LPr{0, {\sigma_a^2}/{n}}$ and ${z}_{i} \stackrel{\rm i.i.d.}{\sim} \N\LPr{0, {\sigma_z^2}}$. Similar to Corollary \ref{cor:k-spaese}, the following corollary characterizes the convergence performance of C-GD combined with the described compression code, when applied to measurements $\mvec{y}$.

\begin{coro}\label{cor:Qpiecepoly}
 Set the step size in C-GD as $\eta={n \over \sigma_a^2 m}$ and the quantization level in the compression code as 
 \[
 b=\lceil (\gamma+0.5)\log n+\log(N+1) \rceil,
 \]  
 where $\gamma>0$ is given.  Set $\tilde{r} =\left((\gamma+0.5)(N+1)(Q+1)+Q\right)\log n+(N+1)(Q+1)(\log(N+1)+1)+1$. Then, given $\epsilon>0$, for  $m\ge 80\tilde{r}\LPr{1+ \epsilon}$,  for $t=0,1,2,\ldots$, we have
\begin{align}\label{eq:noisyfromcor}
\frac{1}{\sqrt{n}}\|\mvec{x}^{k+1} - \tilde{\mvec{x}}\|_2 \le  \frac{0.9}{\sqrt{n}}\|\mvec{x}^k - \tilde{\mvec{x}}\|_2 +  2\LPr{2+\sqrt{\frac{n}{m}}}^2n^{-0.5-\gamma} + \frac{\sigma_z}{\sigma_a}\sqrt{\frac{8(1+\epsilon)\tilde{r}}{m}},
\end{align}
with a probability larger than $1-2^{-2\epsilon \tilde{r}+1}$.
\end{coro}
\begin{proof}
To apply Theorem \ref{noisy_dense_thm_iid_gauss}, we need to find the rate-distortion performance of the described compression code. Using $b$ bits per coefficient, the described code, in total, spends at most $r$ bits, where
\[
r\leq (N+1)(Q+1)b+Q(\log n +1).
\]
For $ b=\lceil (\gamma+0.5)\log n+\log(N+1) \rceil$,
\[
r\leq \left((\gamma+0.5)(N+1)(Q+1)+Q\right)\log n+(N+1)(Q+1)(\log(N+1)+1)+1.
\]
On the other hand, using $b$ bits per coefficient, the distortion in approximating each point can be bounded as
          \begin{align}
              \left|\sum_{i=0}^{N_{\ell}} a^{\ell}_i t^n-\sum_{i=0}^{N_{\ell}} [a^{\ell}_i ]_{b} t^n \right| &\leq \sum_{i=0}^{N_{\ell}} |a^{\ell}_i- [a^{\ell}_i]_{b}|\nonumber\\
              & \leq (N_{\ell}+1) 2^{-b}\leq (N+1) 2^{-b},
              \end{align}
where $[a^{\ell}_i ]_{b}$ denotes the $b$-bit quantized version of $a^{\ell}_i$. Therefore, the overall error is bounded as
\[
\delta\leq \sqrt{n}(N+1) 2^{-b}.
\]
Choosing $ b=\lceil (\gamma+0.5)\log n+\log(N+1) \rceil$, as prescribed by the corollary,  ensures that 
\[
\delta\leq n^{-\gamma}.
\]
Inserting these numbers in  Theorem \ref{noisy_dense_thm_iid_gauss} yields the desired result. 

\end{proof}

The important quantities in the above corollary are explained in the following.
\begin{enumerate}
\item Required number of measurements: If we assume that $n$ is much larger than $N$ and $Q$, then the required number of measurements is $\Omega((N+1)(Q+1) \log n)$. Note that given the degrees of freedom of piecewise polynomial functions, we do not expect to be able to recovery $\mvec{x} \in {\cal P}$ with fewer than $(N+1)(Q+1)$ observations. 

\item Reconstruction error in the absence of measurement noise: Similar to the discussion of the previous section we can argue that 
\[
\lim_{k \rightarrow \infty} \frac{1}{\sqrt{n}}\|\mvec{x}^{k+1} - \tilde{\mvec{x}}\|_2 = O\left( \frac{n^{\frac{1}{2} - \gamma}}{m} \right).
\]
Hence, the error goes to zero for every $\gamma>0.5$.  

\item Reconstruction error in the presence of measurement noise: In this case, the impact of Gaussian noise in the upper bound is $O\left(\frac{\sigma_z}{\sigma_a} \sqrt{\frac{(Q+1)(N+1) \log n}{m}}\right)$. 
\end{enumerate}

\section{Simulation Results and Discussion}\label{sec:simulations}

In this section, we assess the performance of the C-GD algorithm in various settings. Furthermore, we compare our results with the state-of-the-art recovery algorithms, such as  \emph{denoising-based approximate message passing} (D-AMP) \cite{metzler2014denoising} and \emph{nonlocal low-rank regularization} (NLR-CS) \cite{dong2014compressive}. Throughout this section,  when  compression algorithm $\mathcal{X}$ is used in the platform of C-GD, the resulting algorithm is referred to $\mathcal{X}$-GD. For instance, the recovery algorithm that employs the JPEG code is called JPEG-GD. 

\subsection{Parameters setting}\label{sec:simulations_parameter_setting}
Running C-GD involves specifying   three free parameters: (i) step-size $\eta$, (ii) compression rate $r$, and (iii) the number of iterations we run the algorithm. The success of C-GD relies on proper tuning of the first two parameters, \etc  step-size and compression rate. In this section, we  explain how we tune these parameters in practice. In Section \ref{sec:converge_analysis}, we theoretically showed that the algorithm converges to the optimal solution for $\eta = \frac{1}{m \sigma_a}.$
However, this choice of step size might yield a very slow convergence in practice.  Hence, in our simulations we follow an adaptive strategy for setting $\eta$. Let $\eta_k$ denote the step size at the $k^{\rm th}$ iteration. Then, we set $\eta_{k}$ to
\begin{align}
\label{adaptive_alpha_update}
\eta_{k} =  \argmin{\eta}\,\, l\LPr{ \P_{\C_r}\LPr{\mvec{x}^{k} + \eta \,\mmat{A}^{T}\LPr{\mvec{y} - \mmat{A}\mvec{x}^k}}},
\end{align}
where, for $\mvec{u}\in\mathds{R}^n$, $l\LPr{\mvec{u}} \triangleq  \Lp{\mvec{y}-\mmat{A}\mvec{u}}{2}{}$. In other words, $\eta_k$ is set such that the next estimate is moved as close as possible to the subspace $ \mathcal{V} = \{ \mvec{u} \ | \ \mvec{y}= \mmat{A} \mvec{u}\}$. Note that regardless of the value of $\eta$,  the next estimate will be a codeword. Hence, intuitively speaking, the closer this codeword is to the subspace $\mathcal{V}$, a better estimate $\mvec{x}^{k+1}$ will be.  The optimization problem proposed in \eqref{adaptive_alpha_update} is a simple scalar optimization problem, and we use derivative-free methods,
 such as Nelder-Mead method (or downhill simplex) method \cite{nelder1965simplex}, to solve it.  In our simulations we noticed that this strategy speeds up the convergence rate of C-GD. 

Finding the optimal choice of $r$ is an instance of the model selection problem in statistics and machine learning. (See Chapter 7 of \cite{friedman2009elements}.) Hence, standard techniques such as multi-fold cross validation can be used. Note that multi-fold cross validation increases the computational complexity of our recovery algorithm. Reducing the computational complexity of such model selection techniques is left for future research. 

To control the number of iterations in the C-GD method, we consider two standard stopping rules.  One is to limit the maximum number of iterations, which is defined as the parameter $K_{1, \max}$ in Algorithm \ref{euclid}.  The second stopping rule is a predefined threshold on the reduction of squared-error in each iteration, \etc $\Lp{ \mvec{x}^{k+1}-\mvec{x}^{k} }{2}{}$. In our numerical simulation we set $K_{1, \max} = 50$ and $\eps_T = 0.001$. The specific  algorithm that is employed in our simulations is presented below in Algorithm \ref{euclid}. 

\begin{algorithm}\label{alg:1}
\caption{C-GD: Compression-based (projected) gradient descent}\label{euclid}
\begin{algorithmic}[1]
\State {\bf Inputs}: compression code $\LPr{f_r, g_r}$, $\mvec{y}$, $\mmat{A}$
\State {\bf Initialize:} $\mvec{x}^0$, $\eta_{0}$, $K_{1, \max}$, $K_{2,\max}$, $\epsilon_T$
\For{ $k \le K_{1, \max}$}
        \State $\mvec{x}^{k+1} \leftarrow\,   g_r\left(f_r \LPr{\mvec{x}^{k} + \eta_k \,\mmat{A}^{T}\LPr{\mvec{y} - \mmat{A}\mvec{x}^k}}\right)$
        	\State $k \baks k+1$
	  	\State $\eta_{k} \baks$  \text{Apply Nelder-Mead method with maximum} $K_{2, \max}$ \text{iterations to solve} \eqref{adaptive_alpha_update}.
	   \State \If{ $\frac{1}{\sqrt{n}}\Lp{ \mvec{x}^{k+1}-\mvec{x}^{k} }{2}{} < \epsilon_T$} 
	            return $\mvec{x}^{k+1}$
\EndIf
\EndFor
\State{\bf Output:} $\mvec{x}^{k+1}$

\end{algorithmic}
\end{algorithm}

\subsection{Algorithms and comparison criteria}

We explore the performance of our C-GD algorithm for the compressive imaging application. We employ the standard image compression algorithms JPEG and JPEG2000 in our C-GD framework and obtain JPEG-GD and JP2K-GD recovery schemes. In our numerical simulation, we use the  implementation of JPEG2000 and JPEG codecs in the \emph{Matlab}-R2016b Image and Video processing package. 

\begin{figure}[t!]
\begin{center}
\includegraphics[width= 3.5in]{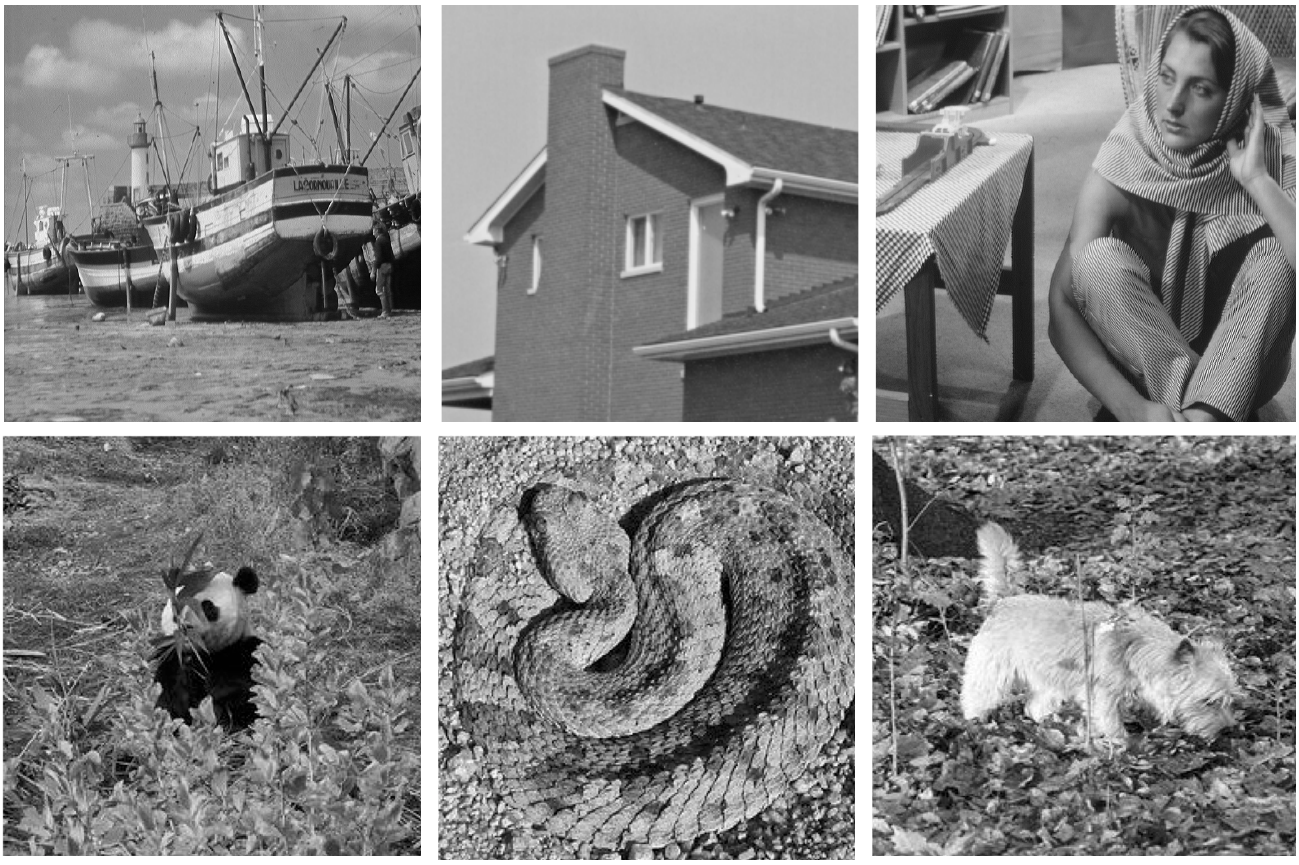}
\caption{Test images used in our simulations.}
\label{fig:TestImages}
\end{center}
\end{figure}

We compare the performance of our algorithm on six standard test images that are shown in Fig.~\ref{fig:TestImages}. To quantitatively evaluate the quality of an estimated image,  we use the peak signal-to-noise ratio (PSNR) defined as 
\begin{align}
{\rm PSNR} = 20\log \LPr{\frac{255}{\sqrt{{\rm MSE}}}},
\end{align}
where for a noise-free grayscale image $\mvec{x}$ and its reconstructed image $\hat{\mvec{x}}$, the mean square error (MSE) is defined as ${\rm MSE} = \frac{1}{n}\Lp{\mvec{x} - \hat{\mvec{x}}}{2}{2}$. 
Furthermore, we define the signal-to-noise ratio (${\rm SNR}$) in the measurement vector $\mvec{y}$ as 
\begin{align}
{\rm SNR} = 20\log  \LPr{\frac{\Lp{\mmat{A}\mvec{x}}{2}{}}{ \Lp{\mvec{z}}{2}{}} }
\end{align}


Even though our theoretical results consider i.i.d.~Gaussian  matrices, we evaluate the performance of our algorithm with both i.i.d.~Gaussian and   partial-DCT measurement matrices, which are closer to the matrices used in radar and magnetic resonance imaging applications. We summarize the results of our simulations for Gaussian matrices and partial-DCT matrices in Sections  \ref{ssec:Gaussiansimulation} and \ref{ssec:FourierSim}, respectively. 

\subsection{Compressive imaging with i.i.d.~measurement matrices}\label{ssec:Gaussiansimulation}

\subsubsection{Noiseless}
In Table \ref{reconstruction_Gaussian}, we compare the results of JPEG-GD and JP2K-GD with that of the state-of-the-art  BM3D-AMP method in reconstructing several test images from their compressive measurements. In these simulations, the  measurement matrices are i.i.d.~Gaussian. Furthermore, the test images are resized to $128\times 128$ pixels. We consider two sampling ratios $\frac{m}{n} = $ $30\%$ and $50\%$. At each iteration, the step-size parameter $\eta$ is set by solving \eqref{adaptive_alpha_update}. To find the solution of this optimization problem we used $K_{2,\max} = 25$ iterations of Nelder-Mead method algorithm. Furthermore, we used the stopping criteria discussed in Section \ref{sec:simulations_parameter_setting}. For BM3D-AMP we also used the default setting proposed in \cite{metzler2014denoising}. 

The results of this simulation are presented in Table \ref{reconstruction_Gaussian}.  Interestingly, results in Table \ref{reconstruction_partialDCT_nonoise} indicate that JP2K-GD considerably outperforms JPEG-GD. The main reason is that JP2K codec exploits more complex structures of natural images compared with JPEG codec.  Therefore, intuitively we can state that for a same amount of measurements,  JP2K-GD can perform better.  Also the performance of JP2K-GD is comparable with that of BM3D-AMP.  When an image has more geometry (as in House), BM3D-AMP outperforms JP2K-GD. However, when an image has more irregular structures and texture such as the Dog image or the Panda image, then JP2K-GD seems to often outperform BM3D-AMP. 

\begin{table}[]
\centering
\caption{PSNR of $128\times 128$ reconstructions with no measurement noise - Sampled by a random Gaussian measurement matrix.}
\label{reconstruction_Gaussian}
{\begin{tabular}{|c|c|c|c|c|c|c|c|}
\hline
Method & \begin{tabular}[c]{@{}c@{}}${m}/{n}$ \end{tabular} 
& Boat  & House & Barbara &  Dog & Panda & Snake \\ \hline \hline \hline 
\multirow{2}{*}{BM3D-AMP} & 30\%                                                                         
& 29.66 & {\bf 39.71} & {\bf 31.3} &  21.30 & 23.90 & {\bf 20.87}\\ \cline{2-8}& 50\%                                                                         
& 34.19 & {\bf 43.70} & {33.70} &  24.35 & 26.76 & 22.76\\
 \hline \hline \hline
\multirow{2}{*}{JPEG-GD}  & 30\%                                                                         
& 23.77& 30.61 & 24.34 &   18.01 & 19.35 & 18.23\\ \cline{2-8} & 50\%                                                                         
& 26.46 & 33.25 & 27.01 &   22.78 & 23.11 & 19.98\\ \hline \hline \hline
\multirow{2}{*}{JP2K-GD}  & 30\%                                                                         
& {\bf 30.68} & 35.22  & 29.96 &  {\bf 22.45} & {\bf 24.00} &  {20.44}  \\ \cline{2-8} & 50\%                                                                         
& {\bf 35.28} & 40.18  & {\bf 34.67} &  {\bf 26.35} & {\bf 27.13}& {\bf 23.03} 
\\ \hline
\end{tabular}}
\end{table}

\subsubsection{Noisy}
In Table \ref{denoising_Gaussian}, we present the performance results of our proposed JPEG-GD and JP2K-GD, and compare them with the performance of the  BM3D-AMP method for image reconstruction from {\em noisy} compressive measurements. Similar to the previous section, the measurement matrix is i.i.d.~ Gaussian,  the images  are resized to $128\times 128$, and two sampling ratios of $30\%$ and $50\%$ are considered.  Unlike before, the measurements are  corrupted by i.i.d.~ Gaussian noise. We consider two different values of ${\rm SNR} = 10\,dB$ and ${\rm SNR} = 30\,dB$. The results of our simulations are presented in Table \ref{denoising_Gaussian}. As is again clear from this table our results are comparable and in most cases  better than the results of the state-of-the-art BM3D-AMP.

\begin{table}[]
\centering
\caption{PSNR of reconstruction of $128\times 128$ test images with Gaussian measurement noise with various SNR values - sampled by a random Gaussian measurement matrix}
\label{denoising_Gaussian}
{\begin{tabular}{cc|c|c|c|c|c|c|}\cline{3-8}
& & \multicolumn{2}{c|}{Barbara} & \multicolumn{2}{c|}{Boat}  &  \multicolumn{2}{c|}{Panda} \\ \hline
\multicolumn{1}{|c|}{Method}    & \multicolumn{1}{l|}{$m/n$} & \multicolumn{1}{l|}{SNR=10} & \multicolumn{1}{l|}{SNR=30} & \multicolumn{1}{l|}{SNR=10} & \multicolumn{1}{l|}{SNR=30} & \multicolumn{1}{l|}{SNR=10} & \multicolumn{1}{l|}{SNR=30} \\ \hline \hline \hline
\multicolumn{1}{|c|}{\multirow{2}{*}{BM3D-AMP}} & 30\%  
& {\bf 19.15}  & {\bf 28.20}  & 20.12  & {\bf 28.50}  &   13.67 & 18.82 \\ \cline{2-8} \multicolumn{1}{|c|}{} & 50\%                                
& {\bf 21.38}  & {     30.16}  & 22.72  &      33.65   &  18.44  & 21.17 \\ \hline
\hline \hline
\multicolumn{1}{|c|}{\multirow{2}{*}{JPEG-CG}}  & 30\%                                
& 14.87 & 22. 71 & 13.50  & 22.01 &  10.50 &  18.79 \\ \cline{2-8} \multicolumn{1}{|c|}{} & 50\%                                
& 18.44 & 24.60  & 19.21  & 25.36 &  15.83 & 22.01 \\ \hline \hline \hline
\multicolumn{1}{|c|}{\multirow{2}{*}{JP2K-CG}}  & 30\%                                
& 16.82   & 26.23 & {\bf 21.79}  & {     28.43} &    {\bf 15.63} & {\bf 22.93} \\ \cline{2-8} \multicolumn{1}{|c|}{} & 50\%                                
& 20.78   & {\bf 30.93} & {\bf 24.82}  & {\bf 34.13} &    {\bf 20.40} & {\bf 25.85} \\ \hline  
\end{tabular}}
\end{table}

\subsection{Compressive imaging with partial-DCT matrices}\label{ssec:FourierSim}
In many application areas, measurement matrices such as partial-DCT matrices are employed. In this section we evaluate the performance of our algorithm on partial-DCT matrices. We note that in our numerical results we observe that even though BM3D-AMP performs well for i.i.d. Gaussian measurements,  its performance degrades dramatically for partial-DCT matrices.  Hence, we compare the performance of our algorithm with the state-of-the-art algorithm for partial-DCT (or partial-Fourier) matrices, \etc NLR-CS \cite{dong2014compressive}. As in the previous section, we consider both noiseless and noisy measurements. 

\subsubsection{Noiseless}

In Table \ref{reconstruction_partialDCT_nonoise}, we present the performance results of our proposed JPEG-GD and JP2K-GD, and compare them with the performance of the NLR-CS method in image reconstruction from compressive samples, sampled by a random partial-DCT measurement matrix. Images in this numerical comparison are resized to $512\times 512$. We consider two sampling ratios $10\%$ and $30\%$.  For $\frac{m}{n} = 10\%$ (sampling rate), JP2K-GD performs comparable and in some cases, \eg Dog and Snake, even better than the  state-of-the-art NLR-CS.   Increasing the sampling ratio to $\frac{m}{n} = 30\%$,  JP2K-GD  outperform both JPEG-GD and NLR-CS methods. 

Note that NLR-CS method has two main steps. In the first step, it estimates an initial image $\hat{\mvec{x}}$ using a standard compressed sensing (CS) recovery method based on the sparsity of image coefficients in DCT/Wavelet domain. Then, in the second step, it enforces a low-rankness and group-sparsity constraint  on the group of similar patches detected in the estimated image \cite{dong2014compressive}.  This step involves   singular value decomposition of a matrix and hence  is computationally expensive for large images. Furthermore, since in the second step, detection of similar patches is performed using an estimated image, exploiting structures  in the second step heavily relies on the performance of the first step. For this particular reason, as observed  in Section \ref{fourier_noisy_numerical}, the performance of NLR-CS method  degrades significantly once noise is added to the observations. On the other hand,  results in Section \ref{fourier_noisy_numerical} show that both JPEG-GD and JP2K-GD are robust to the measurement noise.  

\begin{table}[]
\centering
\caption{PSNR of $512\times 512$ reconstructions with no noise - sampled by a random partial-DCT measurement matrix.}
\label{reconstruction_partialDCT_nonoise}
{\begin{tabular}{|c|c|c|c|c|c|c|c|}
\hline
Method & \begin{tabular}[c]{@{}c@{}}${m}/{n}$ \end{tabular} 
& Boat  & House & Barbara &  Dog & Panda & Snake \\ 
\hline  \hline  \hline
\multirow{2}{*}{NLR-CS} & 10\%                                                                         
& {\bf 23.06} & {\bf 27.26} & {\bf 20.34} & 19.53  & {\bf 21.61} & 18.20\\ \cline{2-8} & 30\%                                                                         
& 26.38 & 30.74 & 23.67 & 23.04 & 25.60 & 21.80\\
\hline  \hline  \hline
\multirow{2}{*}{JPEG-CG}  & 10\%                                                                         
& 18.38 & 24.11 &  16.36 &  16.30 & 17.00 & 15.10 \\ \cline{2-8} & 30\%                                                                         
& 24.70 & 30.51 & 20.37 & 21.10 & 22.01 & 21.63 \\ 
\hline  \hline  \hline
\multirow{2}{*}{JP2K-CG}  & 10\%                                                                         
 & 20.75 & 26.30 & 18.64 & {\bf 19.74}  &  18.24 & {\bf 18.36}\\ \cline{2-8} & 30\%                                                                         
& {\bf 27.73} & {\bf 38.07} & {\bf 24.89} & {\bf 24.82} & {\bf 25.70} & {\bf 24.37} \\
 \hline
\end{tabular}}
\end{table}

\subsubsection{Noisy} \label{fourier_noisy_numerical}
In Table \ref{reconstruction_partialDCT_noisy}, we present the performance results of  JPEG-GD and JP2K-GD, and NLR-CS methods for image reconstruction from {\em noisy} compressive measurements. Similar to the previous section, the measurement matrix is a random partial-DCT matrix. Images in this numerical comparison are resized to $512\times 512$, and we consider two sampling ratios $10\%$ and $30\%$.  The measurements are corrupted by i.i.d.~ Gaussian noise. We consider two different values of ${\rm SNR} = 10\,dB$ and ${\rm SNR} = 30\,dB$. The results of our simulations are presented in Table \ref{reconstruction_partialDCT_noisy}. As is again clear from this table JP2K-GD method outperforms both JPEG-GD and NLR-CS  for all SNRs and sampling ratios ($m/n$). Interestingly we observe that for low SNRs, e.g., SNR = $10$ dB, even the JPEG-GD algorithm performs much better than NLR-CS.

Our next goal is to visually compare the reconstruction of JP2K-GD with that of NLR-CS. Fig.~\ref{fig_codec_visual_comp} shows the reconstructed images for three different sampling ratios $m/n$. The size of the test image for all scenario is $512\times 512$ and the measurement SNR is set to $30$ dB. As is clear from Fig.~\ref{fig_codec_visual_comp}, in  all cases,  the reconstruction from JP2K-GD looks more appealing than the reconstruction from NLR-CS.   
\begin{figure}[htbp]
\begin{center}
\includegraphics[width = 3.5in]{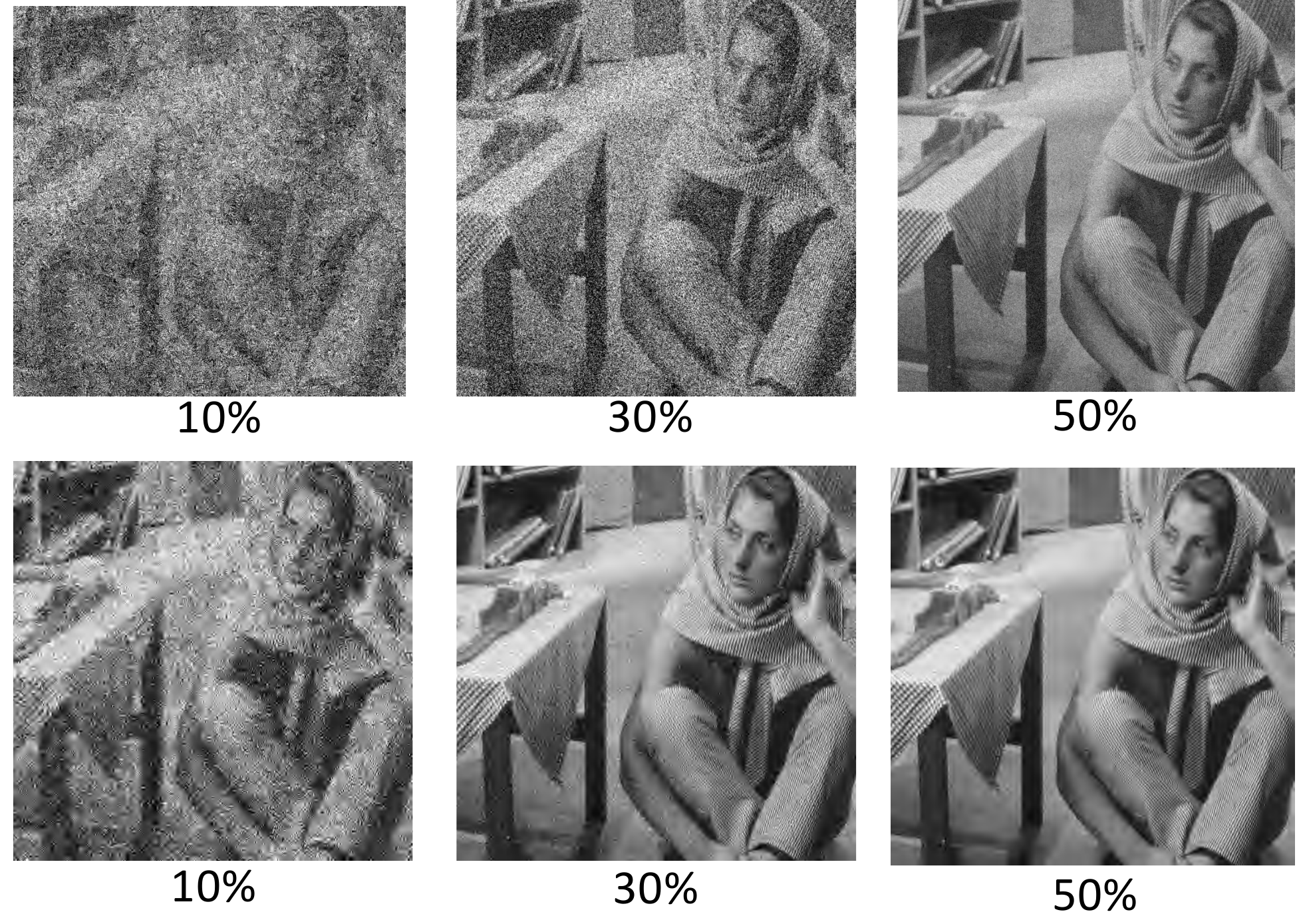}
\caption{Image reconstruction using partial-DCT matrices via using $\frac{m}{n} = 10\%$, $30\%$, and $50\%$   noisy measurements with SNR = 30 dB. The first row illustrates the images reconstructed by NLR-CS method and the second row illustrates reconstructed images by JP2K-GD method. The test image Barbara is resized to $512\times 512$.}
\label{fig_codec_visual_comp}
\end{center}
\end{figure}

\subsection{Convergence rate evaluation}
As proved in our main theorems, we expect the convergence to be linear when the measurement matrix is Gaussian. It turns out that the convergence is also fast for partial-DCT matrices. Fig.~\ref{fig_codec_convergence_nmse} depicts the normalized mean-squared-error (MSE) of image reconstruction using JP2K-GD method. We consider two different sampling ratios $\frac{m}{n} = 5\%$ and $10\%$ in this test. Results in Fig.~\ref{fig_codec_convergence_nmse} and \ref{fig_codec_iteration_quality} show that (i) the algorithm converges very fast (often in less than 50 iterations), (ii) by increasing the number of measurements the convergence of JP2K-GD improves, and  (iii) the final reconstructed image has a better PSNR. Note that all these conclusions are consistent with the results we proved for sub-Gaussian matrices. 
\begin{figure}[htbp]
\begin{center}
\includegraphics[width = 3.5in]{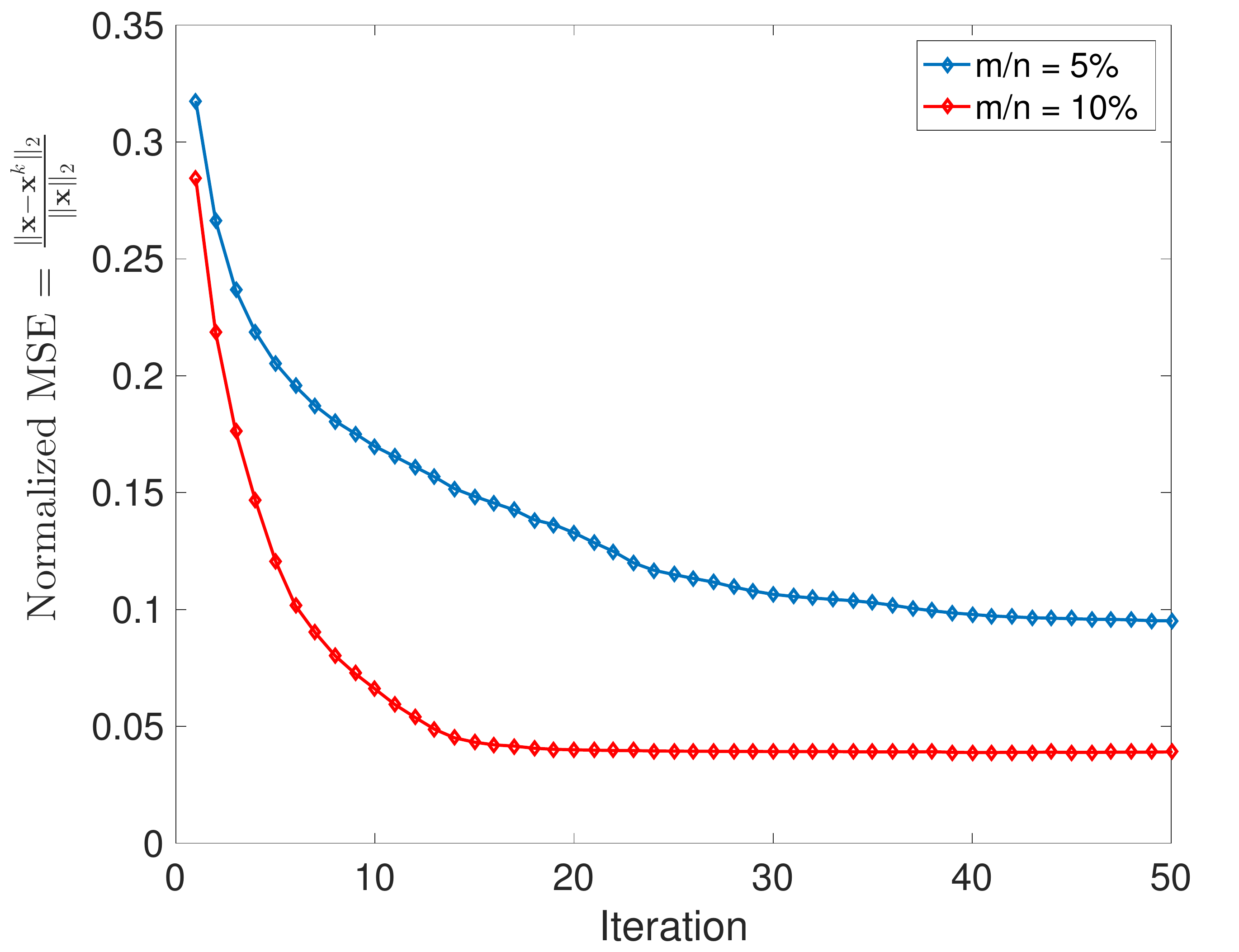}
\caption{Normalized reconstruction error in each iteration of JP2K-GD method on compressive measurements, sampled by a random partial-DCT measurement matrix. House $512 \times 512$-test image.}
\label{fig_codec_convergence_nmse}
\end{center}
\end{figure}

\begin{figure}[htbp]
\begin{center}
\includegraphics[width = 5in]{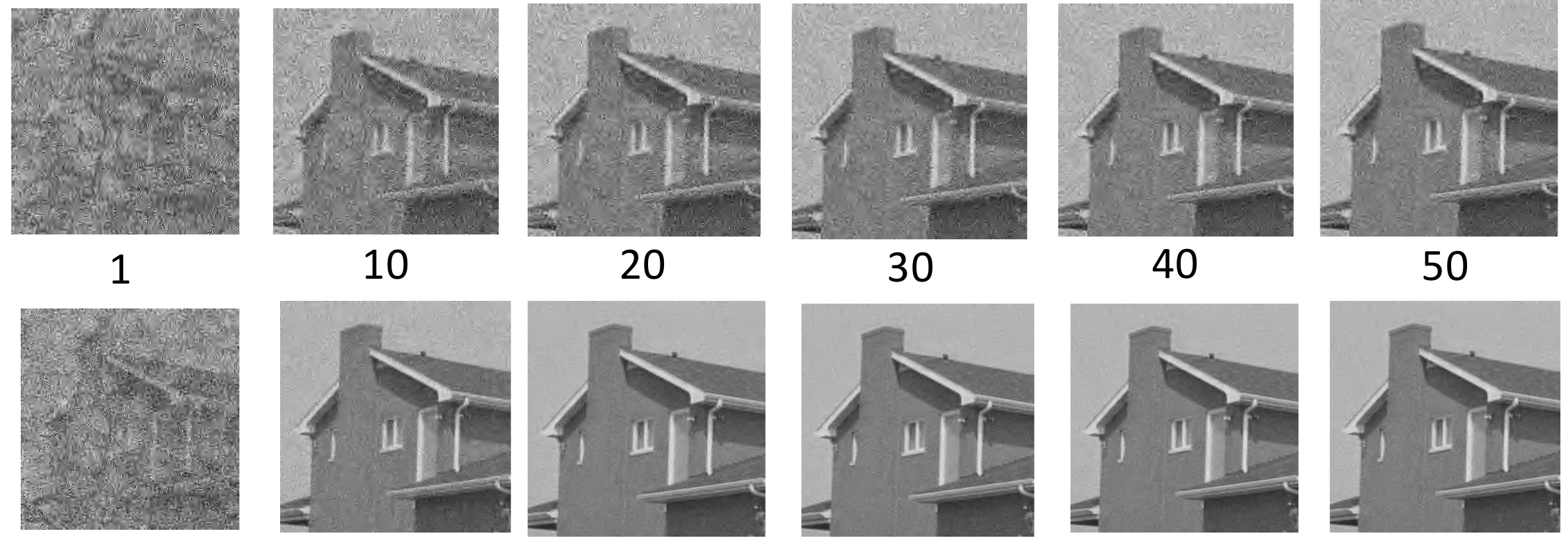}
\caption{Reconstructed image in different itetrations using JP2K-GD method via compressive measurements, sampled by a random partial-DCT measurement matrix. First row of images associated with $m/n = 5\%$ and the second row is associated with $m/n=10\%$ scenario. Numbers between the figures indicate the corresponding iteration number.}
\label{fig_codec_iteration_quality}
\end{center}
\end{figure}

\begin{table}[]
\centering
\caption{PSNR of $512\times 512$ reconstructions with Gaussian measurement noise with various SNR value - sampled by a random partial-DCT measurement matrix.}
\label{reconstruction_partialDCT_noisy}
{\begin{tabular}{cc|c|c|c|c|c|c|}\cline{3-8}
& & \multicolumn{2}{c|}{MRI}  & \multicolumn{2}{c|}{Barbara} & \multicolumn{2}{c|}{Snake} \\ \hline
\multicolumn{1}{|c|}{Method}    & \multicolumn{1}{l|}{$m/n$} & \multicolumn{1}{l|}{SNR=10} & \multicolumn{1}{l|}{SNR=30} & \multicolumn{1}{l|}{SNR=10} & \multicolumn{1}{l|}{SNR=30} & \multicolumn{1}{l|}{SNR=10} & \multicolumn{1}{l|}{SNR=30} \\ \hline  \hline  \hline
\multicolumn{1}{|c|}{\multirow{2}{*}{NLR-CS}} & 10\%  
& 11.66 & 24.14  & 12.10 & 19.83 &  10.50 & 18.75  \\ \cline{2-8} \multicolumn{1}{|c|}{} & 30\%                                
& 12.60 & 26.84 & 13.32 & 24.05 & 11.98 & 24.82  \\ \hline \hline  \hline
\multicolumn{1}{|c|}{\multirow{2}{*}{JPEG-GD}}  & 10\%                                
&14.34 & 20.50 & 15.60 &  18.60 &  12.33 & 15.67   \\ \cline{2-8} \multicolumn{1}{|c|}{} & 30\%                                
&19.20 & 24.70 & 18.17 &  22.89 &  14.40 & 22.37   \\ \hline  \hline  \hline
\multicolumn{1}{|c|}{\multirow{2}{*}{JP2K-GD}}  & 10\%                                
& {\bf 17.33} & {\bf 25.40} & {\bf 16.53} & {\bf 21.65} & {\bf 18.00}  &  {\bf 23.12} \\ \cline{2-8} \multicolumn{1}{|c|}{} & 30\%                                
& {\bf 21.56} & {\bf 35.38} & {\bf 21.82} & {\bf 28.19} & {\bf 21.06} & {\bf 29.30}  \\ \hline 
\end{tabular}}
\end{table}

\section{Proofs} \label{sec:proofs}
\subsection{Background}

In this section we briefly review some useful results that are going to be used in the proofs. 

\begin{lem} [see Lemma 5.9 in \cite{vershynin2010introduction}]
\label{sum_sub-Gaussian}
Let $\LKr{X_i}_{i=1}^n$ be independent, mean zero, sub-Gaussian random variables and $\LKr{a_i}_{i=1}^n$ are real numbers. Then  $\sum_{i=1}^{n} a_iX_i$ is also a sub-Gaussian random variable, and
\begin{align}
\Lp{\sum_{i=1}^{n} a_iX_i}{\psi_2}{} \le \sqrt{\sum_{i=1}^na_i^2 \Lp{X_i}{\psi_2}{2}}.
\end{align}
\end{lem}

%

\begin{thm}[Bernstein Type Inequality, see \eg \cite{vershynin2010introduction}]
\label{thm_bound_sub_exp}
Suppose that  $\LKr{X_i}_{i=1}^n$ are independent, and that, for $i = 1, \cdots, n$,  $X_i$ is a sub-exponential random variable.  Let $\max_i \Lp{X_i}{\psi_1}{}\le K$, for some $K>0$. Then for every $t\ge 0$ and every $\mvec{w} = \LCr{w_1, \cdots, w_n}^T \in \setR^{n\times 1}$, we have
\begin{align}
P\LPr{\sum_{i=1}^n w_i\LPr{X_i- \Eox{X_i}} \ge t} \le \exp\LKr{-\min \LPr{\frac{t^2}{4K^2\Lp{\mvec{w}}{2}{2}}, \frac{t}{2K\Lp{\mvec{w}}{\infty}{} } } } .
\end{align}

\end{thm} 


\begin{lem}[Lemma 3 from \cite{JalaliM:14}] 
\label{same_distribution_lemma}
Consider two independent random vectors $\mvec{X} = [X_1, \cdots, X_n]^T \in \setR^n$ and $\mvec{Y}= [Y_1, \cdots, Y_n]^T \in \setR^{n}$. Assume that $X_i\stackrel{\rm i.i.d.}{\sim}\N(0,1)$ and $Y_i\stackrel{\rm i.i.d.}{\sim}\N(0,1)$.  Then $\LPd{\mvec{X}, \mvec{Y}}$ and $G\Lp{\mvec{X}}{2}{}$ have the same distribution, where $G\sim \N(0,1)$ and is independent of $\Lp{\mvec{X}}{2}{}$.
\end{lem}

\begin{lem}[Lemma 2 from \cite{JalaliM:14}]
\label{X2_concentaration}
Let $G_i$,  $i=1, 2, \cdots, m$,  be i.i.d.~$\N(0,1)$. Then, for $\tau\in(0,1)$,
\[P\LPr{\sum_{i=1}^{m} G_i^2 \le m(1-\tau)} \le \exp\LKr{\frac{m}{2} \LPr{\tau+\ln(1-\tau)}}\]
and for $\tau>0$,
\[P\LPr{\sum_{i=1}^{m} G_i^2 > m(1-\tau)} \le \exp\LKr{-\frac{m}{2} \LPr{\tau-\ln(1-\tau)}}.\]

\end{lem}


\begin{thm}[see \eg \cite{rudelson2014recent}]
\label{bound_SV_densmatrix}
Let $\mmat{A} \in \setR^{m\times n}$ be a  dense random matrix whose entries are i.i.d.~zero-mean Gaussian random variables  with unit variance. Then, $t>0$,
\begin{align}
P\LPr{\sigma_{\max}(\mmat{A})\ge \sqrt{m}+\sqrt{n}+t} \le {\rm e}^{-\frac{t^2}{2}}.
\end{align}
\end{thm}

If we substitute $t \baks t\sqrt{m}$, then Theorem \ref{bound_SV_densmatrix} results in Corollary \ref{sup_sigular_dense}.   
\begin{coro}
\label{sup_sigular_dense}
Let $\mmat{A} \in \setR^{m\times n}$ be a random matrix whose entries are independent,  zero-mean Gaussian random variables  with unit variance. Then
\begin{align}
P\LPr{\sigma_{\max}(\mmat{A})\ge (1+t)\sqrt{m}+\sqrt{n}} \le {\rm e}^{-\frac{mt^2}{2}},
\end{align}
$t>0$ is arbitrary variable.
\end{coro}

\begin{thm}
\label{bound_SV_subgmatrix_thm}
Let $\mmat{A} \in \setR^{m\times n}$ be an i.i.d.~matrix such that $A_{i,j}$ is a zero-mean sub-Gaussian random variable with $\Lp{{A}_{i,j}}{\psi_2}{} \le K$ and $\Eox{{A}_{i,j}^2}=\sigma_a^2$. Then, if $m<n$, for any $t>0$,
\begin{align}
P\LPr{ \sigma_{\max}(\mmat{A})\ge \sqrt{2m\sigma_a^2+12nK(1+t)}} \le {\rm e}^{-3nt},
\end{align}
\end{thm}

\begin{proof} 
Let $\N_\varepsilon$ denote a maximal $\varepsilon$-separated subset of $S^{n-1}$. It is straightforward to show that [Lemma 5.2 from \cite{rudelson2014recent}]
\begin{align}
|\N_\varepsilon|\leq \LPr{1 + \frac{2}{\varepsilon}}^n.\label{bd-epsilon-net}
\end{align}
Consider vector  $\mvec{u} \in S^{n-1}$ that satisfies  $\sigma_{\max}(A)=  \max_{\mvec{u}} \|A\mvec{u}\|$. Let   $\mvec{u}_0 \in \N_{n,\varepsilon}$ be such that $\Lp{\mvec{u} - \mvec{u}_0}{2}{}\le \varepsilon$.  Then, by the triangle inequality, we have
\begin{align}						
 \Labs{\LPd{\mmat{A}\mvec{u}, \mmat{A}\mvec{u}}- \LPd{\mmat{A}\mvec{u}_0, \mmat{A}\mvec{u}_0}}
&=  \Labs{\LPd{\mmat{A}\mvec{u}, \mmat{A}(\mvec{u}-\mvec{u}_0)}- \LPd{\mmat{A}\mvec{u}_0, \mmat{A}(\mvec{u}_0-\mvec{u})}}\nonumber\\
&\le\Labs{\LPd{\mmat{A}\mvec{u}, \mmat{A}(\mvec{u}-\mvec{u}_0)}} +\Labs{ \LPd{\mmat{A}\mvec{u}_0, \mmat{A}(\mvec{u}_0-\mvec{u})}}\nonumber \\
&\le 2(\sigma_{\max}(A))^2 \|\mvec{u}-\mvec{u}_0\|\nonumber\\
&\leq 2\varepsilon(\sigma_{\max}(A))^2.\label{eq:upper-bd-62}
\end{align}
On the other hand, again by the triangle inequality, 
\begin{align}	
 \Labs{\LPd{\mmat{A}\mvec{u}, \mmat{A}\mvec{u}}- \LPd{\mmat{A}\mvec{u}_0, \mmat{A}\mvec{u}_0}} & \ge  \Labs{\LPd{\mmat{A}\mvec{u}, \mmat{A}\mvec{u}}}- \Labs{\LPd{\mmat{A}\mvec{u}_0, \mmat{A}\mvec{u}_0}}\nonumber\\
 &= (\sigma_{\max}(A))^2-  \Labs{\LPd{\mmat{A}\mvec{u}_0, \mmat{A}\mvec{u}_0}}\nonumber\\
 &\ge (\sigma_{\max}(A))^2- \max_{\mvec{x}\in\N_{n,\varepsilon}} \Labs{\LPd{\mmat{A}\mvec{x}, \mmat{A}\mvec{x}}}.\label{eq:lower-bd-63}
\end{align}
Combining \eqref{eq:upper-bd-62} and \eqref{eq:lower-bd-63} yields 
\begin{align}
 (\sigma_{\max}(A))^2\leq  (1 - 2\varepsilon)^{-1}     \max_{\mvec{x} \in \N_{n,\eps}}\Labs{ \LPd{\mmat{A}\mvec{x}, \mmat{A}\mvec{x}}}.\label{eq:sigma-max}
\end{align}
To finish the proof, we need to upper bound $\max_{\mvec{x} \in \N_{n,\eps}} \LPd{\mmat{A}\mvec{x}, \mmat{A}\mvec{x}}$. For a fixed $\mvec{x}\in\N_\varepsilon$, by Theorem \ref{thm_bound_sub_exp},  
\begin{align}
\label{eq_bound_blkd_subGauss_mtx}
		P\LPr{ \frac{1}{m} \Lp{\mmat{A}\mvec{x}}{2}{2}-\sigma_a^2\Lp{\mvec{x}}{2}{2} \ge t }  \le 
		 \exp\LKr{-\min \LPr{\frac{mt^2}{4K^2}, \frac{mt}{2K } } }.
\end{align}
Therefore, 
\begin{align*}
P \LPr{\max_{\mvec{x} \in \N_{\varepsilon}} \LPr{ \frac{1}{m} \Lp{\mmat{A}\mvec{x}}{2}{2}-\sigma_a^2\Lp{\mvec{x}}{2}{2}} \ge t}
&\le \Labs{\N_{\varepsilon}} \exp\LKr{-\min \LPr{\frac{mt^2}{4K^2}, \frac{mt}{2K } } }.\\
\end{align*}
Let $\varepsilon={1\over 4}$. Then, from \eqref{bd-epsilon-net}, 
\begin{align}
\Labs{\N_{\varepsilon}} \exp\LKr{-\min \LPr{\frac{mt^2}{4K^2}, \frac{mt}{2K } } }&\le 9^n \exp\LKr{-\min \LPr{\frac{mt^2}{4K^2}, \frac{mt}{2K } } }\nonumber\\
&= \exp\LKr{-\frac{mt}{2K}\min \LPr{\frac{t}{2K},1 }+n\ln 9 }\nonumber\\
&\le \exp\LKr{-\frac{mt}{2K}\min \LPr{\frac{t}{2K},1 }+3n }.
\end{align}
Substituting $t$ by $6n(1+t)K/m$ and noting that for this value of $t$, if $m<3n$, ${t\over 2K}$ is always larger than 1, it follows that
\begin{align*}
P \LPr{\max_{\mvec{x} \in \N_{\varepsilon}} \LPr{ \frac{1}{m} \Lp{\mmat{A}\mvec{x}}{2}{2}-\sigma_a^2\Lp{\mvec{x}}{2}{2}} \ge {6n(1+t)K\over m}}
&\le \exp\LKr{-3nt}.
\end{align*}
Therefore, in summary, form \eqref{eq:sigma-max}, with probability larger than $1-{\rm e}^{-3nt}$,
\[
\sigma_{\max}(\mmat{A})\leq \sqrt{2m\sigma_a^2+12nK(1+t)}.
\]
\end{proof}
Substituting $t$ by $\frac{m\sigma_a^2}{6nK}$ in Theorem \ref{bound_SV_subgmatrix_thm} results in the following corollary. 

\begin{coro}
\label{bound_SV_matrix_lemma}
Let $\mmat{A} \in \setR^{m\times n}$ be an i.i.d.~matrix such $A_{i,j}$ is a zero-mean sub-Gaussian random variable with $\Lp{{A}_{i,j}}{\psi_2}{} \le K$ and $\Eox{{A}_{i,j}^2}=\sigma_a^2$. Then, for $m<n$, \begin{align}
P\LPr{ \sigma_{\max}(\mmat{A})\ge 2\sigma_a\sqrt{m+\frac{3K}{\sigma_a^2}n}} \le {\rm e}^{-\frac{\sigma_a^2}{2K}m}.
\end{align}
\end{coro}




\begin{lem}[Lemma 5 in \cite{jalali2016new}.]
	\label{thm_dense_Gauss_mtx_Lemma}
	Consider $\mvec{u}, \mvec{v} \in S^{n-1}$ and dense random Gaussian matrix $\mmat{A} \in \setR^{m\times n}$ with i.i.d zero mean Gaussian entries as $\N\LPr{0, \sigma_a^2}$. Then, for any $t>0$
	\begin{align}
	\label{eq_concent_dense_gauss}
		P\LPr{ \LPd{\mvec{u}, \mvec{v}} - \frac{1}{m\sigma_a^2}\LPd{\mmat{A}\mvec{u}, \mmat{A}\mvec{v} }  \ge t } \le   {\rm e}^{-m f^*(t)}   \end{align}
	where 
$
		 f^*(t)  = \min_{u \in [-1,1]} \max_{s \in \LPr{0, \frac{1}{1-u} }} 
		  \LKr{s(t-u)+\frac{1}{2}\ln\LCr{ \LPr{1+su}^2 - s^2}}.$
\end{lem}

\begin{coro}
	\label{lem_dense_Gauss_mtx}
	Consider $\mvec{u}, \mvec{v} \in S^{n-1}$ and dense random Gaussian matrix $\mmat{A} \in \setR^{m\times n}$ with i.i.d. zero mean Gaussian entries as $\N\LPr{0, \sigma_a^2}$. Then, 
	\begin{align}
		P\LPr{ \LPd{\mvec{u}, \mvec{v}} - \frac{1}{m\sigma_a^2}\LPd{\mmat{A}\mvec{u}, \mmat{A}\mvec{v} }  \ge 0.45 } \le   2^{-\frac{m}{20}} .  \end{align}
\end{coro}

\begin{lem}
	\label{thm_dense_subGauss_mtx}
	Consider $\mvec{u}, \mvec{v} \in S^{n-1}$ and dense matrix $\mmat{A} \in \setR^{m\times n}$ with i.i.d. zero-mean sub-Gaussian entries with $\Lp{\mmat{A}_{i,j}}{\psi_2}{} \le K$ and $\Eox{A_{i,j}^2}=\sigma_a^2$. Then, for any $t>0$,
	\begin{align}
	\label{eq_concent_dense_subgauss}
		P\LPr{ \LPd{\mvec{u}, \mvec{v}} - \frac{1}{m\sigma_a^2}\LPd{\mmat{A}\mvec{u}, \mmat{A}\mvec{v} } \ge t }  \le 
        \exp\LKr{-\frac{mt\sigma_a^2}{2K^2} \min \LPr{1, \frac{t\sigma_a^2}{2K^2}} }. 
\end{align}
\end{lem}
\begin{proof}
Define $\mvec{y}(u) = \mmat{A}\mvec{u}$ and $\mvec{y}(v) = \mmat{A}\mvec{v}$. Using these definitions, $\LPd{\mmat{A}\mvec{u}, \mmat{A}\mvec{v}}  =\LPd{\mvec{y}(u), \mvec{y}(v)}$. 
Let $\mmat{A}_{i} \in \setR^{1\times n}$ denote the $i$-the row of matrix $\mmat{A}$.  Thus, $\mvec{y}_i(u) = \LPd{\mmat{A}_{i}, \mvec{u}}$ and  $\mvec{y}_{i}(v) = \LPd{\mmat{A}_{i}, \mvec{v}}$ are both sub-Gaussian random variables. Using Lemma \ref{sum_sub-Gaussian}, we have 
$$
\Lp{\mvec{y}_{i}(u) }{\psi_2}{} \le {K}\Lp{\mvec{u}}{2}{} = K, \ \ \ \  \Lp{\mvec{y}_{i}(v) }{\psi_2}{} \le K\Lp{\mvec{v}}{2}{} = K.
$$
Furthermore, $\Eox{\mvec{y}_{i}(u) \mvec{y}_{i}(v)} =  \mvec{u}^T\Eox{\mmat{A}_{i}^T \mmat{A}_{i}} \mvec{v} =\sigma_a^2 \LPd{\mvec{u}, \mvec{v}}.$ Note that 
\begin{align}
		P\LPr{ \LPd{\mvec{u}, \mvec{v}} - \frac{1}{m\sigma_a^2 }\LPd{\mmat{A}\mvec{u}, \mmat{A}\mvec{v} } \ge t }   &=
        		P\LPr{ \sum_{i=1}^m \LPr{\sigma_a^2 \LPd{\mvec{u}, \mvec{v}} - \mvec{y}_{i}(u) \mvec{y}_{i}(v)} \ge m t \sigma_a^2 }  
\end{align}
By Lemma \ref{prod_subgaus}, $\mvec{y}_{i}(u) \mvec{y}_{i}(v)$ is a sub-exponential random variable with 
\[
\Lp{{\mvec{y}}_{i}(u){\mvec{y}}_{i}(v)}{\psi_1}{} \le \Lp{{\mvec{y}}_{i}(u)}{\psi_2}{} \Lp{{\mvec{y}}_{i}(v)}{\psi_2}{} \leq K^2.
\]
Therefore,  by applying Theorem \ref{thm_bound_sub_exp} to sub-exponential random variables $\mvec{y}_{i}(u) \mvec{y}_{i}(v)$, and setting all weights equal to $-1$, we derive 
\begin{align} \nonumber
		P\LPr{ \LPd{\mvec{u}, \mvec{v}} - \frac{1}{m\sigma_a^2 }\LPd{\mmat{A}\mvec{u}, \mmat{A}\mvec{v} } \ge t }  &\leq 
		    \exp\LKr{- \min \LPr{{mt^2\sigma_a^4 \over 4K^4},{mt\sigma_a^2  \over 2K^2}} } . 
\end{align}

\end{proof}

\subsection{Proof of Theorem \ref{noisy_dense_thm_iid_gauss}}
\label{proof_noisy_dense_thm_iid_gauss}
Define  \begin{align}
	\label{grad_dir_upd_dens_gaus}
	\mvec{s}^{k+1} = \mvec{x}^{k} + \eta \,\mmat{A}^{T}\LPr{\mvec{y} - \mmat{A}\mvec{x}^k}.
	\end{align}
Using this notation, we have $ \mvec{x}^{k+1}=\P_{\C}(\mvec{s}^{k+1} )$. But since  $ \tilde{\mvec{x}} =\P_{\C}(\mvec{x})$, $\tilde{\mvec{x}}$ is also in $\C$. Hence,  $\Lp{\mvec{x}^{k+1}-\mvec{s}^{k+1}}{2}{2} \le \Lp{\tilde{\mvec{x}}-\mvec{s}^{k+1}}{2}{2},$  
		 or, equivalently,
	$\Lp{\LPr{\mvec{x}^{k+1}- \tilde{\mvec{x}}} - \LPr{\mvec{s}^{k+1} - \tilde{\mvec{x}}}}{2}{2} \le \Lp{\tilde{\mvec{x}}-\mvec{s}^{k+1}}{2}{2}
	$. By removing the common terms from both sides, we have
	\begin{align}
		\Lp{\mvec{x}^{k+1}- \tilde{\mvec{x}}}{2}{2} \le 2 \LPd{\mvec{x}^{k+1} - \tilde{\mvec{x}}, \mvec{s}^{k+1} - \tilde{\mvec{x}}}  \label{eq:error-step1-42}
	\end{align}
For $k=0,1,\ldots$, define the error vector and its normalized version  as 
	\[
	\mvec{\theta}^{k} \triangleq \mvec{x}^{k}- \tilde{\mvec{x}},
	\]
	and 
		\[
	\underline{\mvec{\theta}}^{k} \triangleq {	\mvec{\theta}^{k}\over \|	\mvec{\theta}^{k}\|},
	\]
respectively. Also, given $\mvec{\theta}^{k}\in\setR^n$, $\mvec{\theta}^{k+1}\in\setR^n$, $\eta\in\setR^+$,  and $\mmat{A}\in\setR^{m\times n}$, define coefficient $\mu$ as
\[
\mu\LPr{\mvec{\theta}^{k+1}, \mvec{\theta}^{k}, \eta} \triangleq  \LPd{\underline{\mvec{\theta}}^{k+1}, \underline{\mvec{\theta}}^{k}} - \eta \LPd{\mmat{A}\underline{\mvec{\theta}}^{k+1}, \mmat{A}\underline{\mvec{\theta}}^{k} }.
\] 
	Using this definition, substituting for $\mvec{s}^{k+1} $ from  \eqref{grad_dir_upd_dens_gaus} and noting that $\mvec{y} = \mmat{A}\mvec{x} + \mvec{z}$, from  \eqref{eq:error-step1-42}, it follows that 
\begin{align}
		\nonumber
		\Lp{\mvec{\theta}^{k+1}}{2}{} &\le 
		\LPr{2 \LPd{\mvec{x}^{k+1} - \tilde{\mvec{x}}, \mvec{x}^{k} + 
		\eta \,\mmat{A}^{T}\LPr{\mmat{A}\mvec{x} + \mvec{z} - \mmat{A}\mvec{x}^{k}}\ - \tilde{\mvec{x}}} }  
		\Lp{\mvec{\theta}^{k+1}}{2}{-1}
		\\\nonumber &= 
		\LPr{2 \LPd{\mvec{x}^{k+1} - \tilde{\mvec{x}}, \mvec{x}^{k} - \tilde{\mvec{x}}} + 
		2\eta \, \LPd{\mvec{x}^{k+1} - \tilde{\mvec{x}}, \mmat{A}^{T}\mmat{A}\LPr{\mvec{x} - \mvec{x}^{k}}} +
		2\eta \, \LPd{\mvec{x}^{k+1} - \tilde{\mvec{x}}, \mmat{A}^{T}\mvec{z}}}\Lp{\mvec{\theta}^{k+1}}{2}{-1}
		\\ \nonumber& = 
		\LPr{2\LPd{\mvec{\theta}^{k+1}, \mvec{\theta}^{k}} - 2\eta \LPd{\mmat{A}\mvec{\theta}^{k+1}, \mmat{A}
		\mvec{\theta}^{k} } + 
		2\eta \, \LPd{\mmat{A}\mvec{\theta}^{k+1}, \mmat{A}\LPr{\mvec{x}-\tilde{\mvec{x}}}}+
		2\eta \, \LPd{ \mvec{\theta}^{k+1}, \mmat{A}^{T}\mvec{z}}}\Lp{\mvec{\theta}^{k+1}}{2}{-1}
		\\ \label{Eq_noisy_error_terms} & \le  
		2\mu\LPr{\underline{\mvec{\theta}}^{k+1}, \underline{\mvec{\theta}}^{k}, \eta}\Lp{\mvec{\theta}^{k}}{2}{} + 
		2\eta \|\mmat{A}\|_{S^{n-1}}^2\Lp{\mvec{x}-\tilde{\mvec{x}}}{2}{} +
		2\eta \LPd{\underline{\mvec{\theta}}^{k+1} , \mmat{A}^T\mvec{z}},
	\end{align}
where $ \|\mmat{A}\|_{S^{n-1}} = \sigma_{\max}\LPr{\mmat{A}}$. We next find upper bounds for the three terms on the right hand side of \eqref{Eq_noisy_error_terms}:
\begin{enumerate}
\item[(i)] Bounding $\mu\LPr{\underline{\mvec{\theta}}^{k+1}, \underline{\mvec{\theta}}^{k}, \eta}$: We show that given the parameter setting of the theorem,  with high probability
\begin{align}
\mu\LPr{\mvec{u}, \mvec{v}, \eta} \le 0.45,\,\,\,\, \text{for}\,\, \forall \, \mvec{u}, \mvec{v} \in \C'
\end{align}
where 
\begin{align}
\C'\triangleq  \LKr{\frac{\hat{\mvec{x}}_1-\hat{\mvec{x}}_2}{\Lp{\hat{\mvec{x}}_1-\hat{\mvec{x}}_2}{2}{}}: \forall \;\hat{\mvec{x}}_1, \hat{\mvec{x}}_2\in \C}.\label{eq:set-C-p}
\end{align}
 To achieve this goal,  we  define event $\E_1$ as 
\begin{align}
	\label{first_term_bound}
\E_1 \triangleq \LKr{ \mu\LPr{\mvec{u}, \mvec{v}, \frac{1}{m\sigma_a^2}} < 0.45  :  \forall \;\mvec{u}, \mvec{v} \in \C'}.
\end{align}
From Corollary \ref{lem_dense_Gauss_mtx} (or Lemma \ref{thm_dense_Gauss_mtx_Lemma}), given $\mvec{u}, \mvec{v}\in\C'$, we have
\begin{align}\label{eq:boundmuthm2}
 P\LPr{\mu\LPr{\mvec{u}, \mvec{v}, \frac{1}{m\sigma_a^2}} \ge 0.45} &\le
2^{-\frac{m}{20}}.
\end{align}
Therefore, by the union bound, 
\begin{align}
 P \LPr{\E_1^c} \le |\C'|^2 2^{-\frac{m}{20}}.
\end{align}
Note that  $ |\C'|\leq  |\C|^2 \leq 2^{2r}$.  Therefore,  
$$
P\LPr{\E_1} \ge 1-|\C'|^2 2^{-\frac{m}{20}} \ge 1 - 2^{\LPr{4r-0.05m}}.
$$
Therefore, for $m\ge 80r\LPr{1 + \epsilon}$, where $\epsilon>0$, with probability at least $1-2^{-40r\epsilon}$, event $\E_1$ happens.

\item[(ii)] Bounding $\|\mmat{A}\|_{S^{n-1}}^2\Lp{\mvec{x}-\tilde{\mvec{x}}}{2}{}$: Define event $\E_2$ as 
\[
\E_2\triangleq \LKr{\sigma_{\max}(\mmat{A})\le 2\sqrt{m}+\sqrt{n}}.
\]
From Corollary \ref{sup_sigular_dense}, for $t=1$ we have
\[
P(\E_2^c)\leq {\rm e}^{-\frac{m}{2}}.
\]
Also, since the compression code has supremum distortion $\delta$, $\Lp{\mvec{x}-\tilde{\mvec{x}}}{2}{} \le \delta$. Therefore, conditioned on $ \E_2$, we have 
\begin{align}\label{eq:secondtermthm2}
\frac{2}{m}\LPr{\sigma_{\max}(\mmat{A})}^2 \Lp{\mvec{x}-\tilde{\mvec{x}}}{2}{} 
\le \frac{2}{m}\LPr{2\sqrt{m}+\sqrt{n}}^2 \delta
= 2\LPr{2+\sqrt{\frac{n}{m}}\;}^2\delta.
\end{align}

\item [(iii)] Bounding $2\eta \LPd{\underline{\mvec{\theta}}^{k+1} , \mmat{A}^T\mvec{z}}$: Note that $2\eta\LPd{\underline{\mvec{\theta}}^{k+1},  \mmat{A}^T\mvec{z}} = \frac{2}{m\sigma_a^2}\LPd{\mmat{A}\underline{\mvec{\theta}}^{k+1},  \mvec{z}}$. Let $\mmat{A}_i \in \setR^n$ be the $i$-th row of matrix $\mmat{A}$. Then, 
$$
\mmat{A}\underline{\mvec{\theta}}^{k+1} =\LCr{\LPd{\mmat{A}_1, \underline{\mvec{\theta}}^{k+1}}, \LPd{\mmat{A}_2, \underline{\mvec{\theta}}^{k+1}},  \cdots, \LPd{\mmat{A}_n, \underline{\mvec{\theta}}^{k+1}} }^T.
$$ 
For any fixed $\underline{\mvec{\theta}}^{k+1}$, $ \left\{ \LPd{\mmat{A}_i, \underline{\mvec{\theta}}^{k+1}} \right\}_{i=1}^n$ are i.i.d.~$\N(0,1)$ random variables. Hence, from Lemma \ref{same_distribution_lemma}, we know that the distribution of $\LPd{\underline{\mvec{\theta}}^{k+1},  \mmat{A}^T\mvec{z}}$ is the same as $  \sigma_a\Lp{\mvec{z}}{2}{} \LPd{\underline{\mvec{\theta}}^{k+1}, \mvec{g}}$, where  $\mvec{g} = \LCr{g_1, \cdots, g_n}^T$ is independent of $\Lp{\mvec{z}}{2}{}$ and $g_i \stackrel{i.i.d.}{\sim} \N(0, 1)$. To bound $  \sigma_a\Lp{\mvec{z}}{2}{} \LPd{\underline{\mvec{\theta}}^{k+1}, \mvec{g}}$ we will bound  $\frac{1}{\sigma_z^2}\Lp{\mvec{z}}{2}{2} $ and $\left|\LPd{\underline{\mvec{\theta}}, \mvec{g}}\right|^2$ separately. Given $\tau'_1>0$ and $\tau'_2>0$, define events $\E_3$ and $\E_4$ as follows
\[
\E_3\triangleq \LKr{\frac{1}{\sigma_z^2}\Lp{\mvec{z}}{2}{2} \leq (1+\tau'_1)m} 
\]
and
\[
\E_4\triangleq \LKr{\left|\LPd{\underline{\mvec{\theta}}, \mvec{g}}\right|^2 \leq 1+\tau'_2, \forall \; \underline{\mvec{\theta}}\in\C'}.
\]
Following Lemma \ref{X2_concentaration}, we have
\begin{align}
\label{step_one_Gaussian_noise}
 P\LPr{\E_3^c} \le {\rm e}^{-\frac{m}{2}\LPr{\tau'_1-\ln\LPr{1+\tau'_1}}},
 \end{align}
and letting $m=1$ in Lemma \ref{X2_concentaration}, for fixed $\underline{\mvec{\theta}}^{k+1}$, it follows that
\begin{align}
 P\LPr{\left|\LPd{\underline{\mvec{\theta}}^{k+1}, \mvec{g}}\right|^2 \ge 1+\tau'_2} \le {\rm e}^{-\frac{1}{2}\LPr{\tau'_2-\ln\LPr{1+\tau'_2}}}.
 \end{align}
Hence, by the union bound,
 \begin{align}
 P\LPr{\E_4^c} \le |\C'| {\rm e}^{-\frac{\tau'_2}{2}} = 2^{2r}{\rm e}^{-\frac{1}{2}\LPr{\tau'_2-\ln\LPr{1+\tau'_2}}} \le 2^{2r-\frac{\tau'_2}{2}},
 \end{align}
where the last inequality holds for $ \tau'_2>7$. Setting $\tau'_2 = 4(1+\epsilon)r-1$, where $\epsilon>0$, ensures that $P\LPr{\E_4^c} \leq  2^{-2\epsilon r+0.5}$.
Setting $\tau'_1 = 1$,  $\P(\E_3^c)\leq {\rm e}^{-0.15 m}$, and conditioned on $\E_3\cap\E_4$, we have
 \begin{align}
2 \eta \LPd{\underline{\mvec{\theta}}^{k+1} , \mmat{A}^T\mvec{z}} &= 
\frac{2}{m\sigma_a} \LPd{\underline{\mvec{\theta}}^{k+1} , \mmat{A}^T\mvec{z}}\nonumber \\
&\le\frac{2}{m\sigma_a} \sqrt{ \sigma_z^2(1+\tau'_1)m(1+\tau'_2)}\nonumber\\
\label{last_step_Gaussian_noise}
&= \frac{2\sigma_z}{m\sigma_a} \sqrt{8m\LPr{1+\epsilon }r} =  \frac{\sigma_z}{\sigma_a}\sqrt{\frac{8(1+\epsilon )r}{m}}.
\end{align}
\end{enumerate}
Combining \eqref{last_step_Gaussian_noise}, \eqref{eq:secondtermthm2}, and \eqref{eq:boundmuthm2} with  \eqref{Eq_noisy_error_terms}  yields the desired bound on the reduction of error. Finally, note that, by the union bound,
\begin{align}
P(\E_1\cap\E_2\cap\E_3\cap\E_4)&\geq 1- \sum_{i=1}^4 P(\E_i)\geq 1-{\rm e}^{-{m\over 2}}-2^{-40r\epsilon} -2^{-2\epsilon r+0.5}-{\rm e}^{-0.15 m}\nonumber\\
&\geq 1-2^{-2\epsilon r+1}.
\end{align}

\subsection{Proof of Theorem \ref{thm:imperfectproject}}\label{ssec:proof:thm:imperfectproject}

The proof of this result is a simple extension of the proof of Theorem \ref{noisy_dense_thm_iid_gauss} presented in Section \ref{proof_noisy_dense_thm_iid_gauss}. Define
\[
\mvec{s}^{k+1} = \mvec{x}^k+ \eta \mmat{A}^T (\mvec{y} - \mmat{A} \mvec{x}^k). 
\]
Note that $\mvec{x}^{k+1} = g_r(f_r(\mvec{s}^{k+1}))$. Hence,
\begin{equation}\label{eq:proofimperfect1}
\|\mvec{x}^{k+1} -\tilde{ \mvec{x}} \|_2 = \|\mvec{x}^{k+1} -{\cal P}_{\cal S} (\mvec{s}^{k+1}) \|_2+ \|{\cal P}_{\cal S} (\mvec{s}^{k+1})-\tilde{ \mvec{x}} \|_2 \leq \|{\cal P}_{\cal S} (\mvec{s}^{k+1})- \tilde{ \mvec{x}} \|_2  + \xi
\end{equation}
 It is now straightforward to follow exactly the same step as the one discussed in the proof of Theorem \ref{noisy_dense_thm_iid_gauss} and show that with probability at least $1- 2^{-4r\epsilon}-{\rm e}^{-{m\over 4}}- 2^{-2r\epsilon}$
\begin{equation}\label{eq:proofimperfect2}
\frac{1}{\sqrt{n}} \|{\cal P}_{\cal S} (\mvec{s}^{k+1})-\tilde{ \mvec{x}} \|_2 \leq \frac{0.9}{\sqrt{n}}\|\mvec{x}^k - \tilde{\mvec{x}}\|_2 +  2\LPr{2+\sqrt{\frac{n}{m}}}^2{\delta \over \sqrt{n}} + \frac{\sigma_z}{\sigma_a}\sqrt{\frac{8(1+\epsilon)r}{m}}.
\end{equation}
Combining \eqref{eq:proofimperfect1} and \eqref{eq:proofimperfect2} completes the proof.

\subsection{Proof of Theorem \ref{noisy_dense_thm_iid_sub}}
\label{proof_noisy_dense_thm_iid_sub}
Following the proof of Theorem \ref{noisy_dense_thm_iid_gauss}, and defining $\mvec{\theta}^{k+1} = \mvec{x}^{k+1}- \tilde{\mvec{x}}$, for $k=0,1,2,\ldots$, it follows from  \eqref{Eq_noisy_error_terms} that
\begin{align}
		\Lp{\mvec{\theta}^{k+1}}{2}{} &\le  
		2\mu\LPr{\underline{\mvec{\theta}}^{k+1}, \underline{\mvec{\theta}}^{k}, \eta} \Lp{\mvec{\theta}^{k}}{2}{} + 
		2\eta \sigma_{\max}^2\LPr{\mmat{A}}\Lp{\mvec{x}_o-\tilde{\mvec{x}}}{2}{} +
		2\eta \LPd{\underline{\mvec{\theta}}^{k+1} , \mmat{A}^T\mvec{z}},\label{eq:th3-main-bound}
	\end{align}
where $\mu\LPr{\underline{\mvec{\theta}}^{k+1}, \underline{\mvec{\theta}}^{k}, \eta} \triangleq \LPd{\underline{\mvec{\theta}}^{k+1}, \underline{\mvec{\theta}}^{k}} - \eta \LPd{\mmat{A}\underline{\mvec{\theta}}^{k+1}, \mmat{A}\underline{\mvec{\theta}}^{k} }$. Let $\C'$ denote the set of normalized distance vectors of the codewords in $\cal C$, which is defined in \eqref{eq:set-C-p}. Define event  $\E_1$ as 
\[
\E_1\triangleq\LKr{\mu\LPr{\mvec{u}, \mvec{v}, \frac{1}{m\sigma_a}} \le \mu_0 :  \forall \mvec{u}, \mvec{v} \in \C' }.
\]
Similar to the proof of Theorem \ref{noisy_dense_thm_iid_gauss}, we show that the probability of occurrence of  $\E_1^c$ approaches  $0$.  
%
Given $ \mvec{u}, \mvec{v} \in \C'$, from Lemma \ref{thm_dense_subGauss_mtx}, we have
	\begin{align}
	\label{eq_concent_dense_subgauss}
		P\LKr{\mu\LPr{\mvec{u}, \mvec{v}, \frac{1}{m\sigma_a^2}} \ge \mu_0} \le   \exp\LKr{-\frac{m\mu_0\sigma_a^2}{2K^2} \min \LPr{1, \frac{\mu_0\sigma_a^2}{2K^2}} } = 2^{-(\log {\rm e})\LPr{\frac{m\mu_0 \sigma_a^2}{2K^2} \min \LPr{1, \frac{\mu_0 \sigma_a^2}{2K^2}} }}.
\end{align}
Therefore, by the union bound, since $ |\C'| \le  |\C|^2 = 2^{2r}$, we have 
$$
P\LPr{\E_1} \ge 1-|\C'|^2 2^{-(\log {\rm e})\LPr{\frac{m\mu_0 \sigma_a^2}{2K^2} \min \LPr{1, \frac{\mu_0\sigma_a^2}{2K^2}} } } \ge 1 - 2^{\LPr{4r-(\log {\rm e})\LPr{\frac{m\mu_0\sigma_a^2}{2K^2} \min \LPr{1, \frac{\mu_0\sigma_a^2}{2K^2}} }}}.
$$
Therefore, for 
$m >\frac{8K^2r(1+\epsilon)}{ \mu_0  \sigma_a^2 \min \LPr{1, \frac{\mu_0 \sigma_a^2}{2K^2}  } \log {\rm e}}$, 
\[
P(\E_1^c)\leq 2^{-4r\epsilon}.
\]
But, since by assumption $\mu_0 \sigma_a^2\leq 2K^2$, we have $ \min \LPr{1, \frac{\mu_0 \sigma_a^2}{2K^2}  }=\frac{\mu_0 \sigma_a^2}{2K^2} $.
\newline
Define event $\E_2$ as $\E_2\triangleq \LKr{\sigma_{\max}(\mmat{A})\le 2\sigma_a\sqrt{m + n\frac{3K}{\sigma_a^2}}}$.  From Corollary \ref{bound_SV_matrix_lemma},  we have
 $P\LPr{\E_2}\geq 1-{\rm e}^{-\frac{m\sigma_a^2}{2K}}$. Since the compression code is such that  $ \Lp{\mvec{x}-\tilde{\mvec{x}}}{2}{}\leq \delta $, conditioned on $\E_2$, we have
\begin{align}
\nonumber
\frac{2}{m\sigma_a^2}\LPr{\sigma_{\max}(\mmat{A})}^2 \Lp{\mvec{x}-\tilde{\mvec{x}}}{2}{} 
&\le \frac{8}{m}\LPr{m+\frac{3K}{\sigma_a^2}n} \delta\nonumber\\
&= 8\LPr{1+\frac{3Kn}{\sigma_a^2m}}\delta.
\end{align}

To complete the proof, we need  to bound $2\eta\LPd{\underline{\mvec{\theta}}^{k+1},  \mmat{A}^T\mvec{z}} = \frac{2}{m\sigma_a^2}\LPd{\mmat{A}\underline{\mvec{\theta}}^{k+1},  \mvec{z}}$, which is the term related to the  noise $\mvec{z}$. Again, let $\mmat{A}_i \in \setR^{1\times n}$ denote the $i$-th row of matrix $\mmat{A}$ and, for a given $\mvec{u}\in\setR^n$, let  $
\mvec{y}_{\mvec{u}} = \mmat{A}{\mvec{u}}.$ Hence, for $i=1,\ldots,m$,
$$
\mvec{y}_{\mvec{u}}(i) = \LPd{\mmat{A}_i, \mvec{u} }.
$$ 
To  upper bound  the term corresponding to noise, given $\tau>0$, define event $\E_3$ as
$$
\E_3 = \LKr{\frac{2}{m\sigma_a^2}\LPd{\mvec{y}_{\mvec{u}},  \mvec{z}} \le \tau: \;\forall \;\mvec{u} \in\C'}.
$$
From Lemma \ref{sum_sub-Gaussian}, for $\mvec{u}\in\C'$, we know that $\LKr{\mvec{y}_{\mvec{u}}(i)}_{i=1}^m$ are independent sub-Gaussian random variables. Also, for $\mvec{u}\in\C'$, Lemma \ref{sum_sub-Gaussian} states that 
$$
\Lp{\mvec{y}_{\mvec{u}}(i)}{\psi_2}{}\le  \Lp{\underline{\mvec{u}}}{2}{} \max_{1\le j \le n} \LPr{\Lp{\mmat{A}_i(j)}{\psi_2}{}}\le K,
$$ 
where the last inequality follows because for $\mvec{u}\in\C'$, $\Lp{\mvec{u}}{2}{}=1$. 
Since every Gaussian random variable is also a sub-Gaussian random variable, $\mvec{z}(i)$ is a sub-Gaussian random variable with $\Lp{\mvec{z}(i)}{\psi_2}{} = \sigma_n\sqrt{8\over 3}$. As a result, 
$\Lp{\mvec{y}_{\mvec{u}}(i)\mvec{z}(i)}{\psi_1}{} \le \Lp{\mvec{y}_{\mvec{u}}(i)}{\psi_2}{}\Lp{\mvec{z}(i)}{\psi_2}{} \le K\sqrt{8\over 3}\sigma_n$. Using Theorem \ref{thm_bound_sub_exp}, for $\mvec{u}\in\C'$, we have
\begin{align*}
P\LPr{\frac{2}{m\sigma_a^2}\LPd{\mvec{y}_{\mvec{u}},  \mvec{z}} \ge \tau}&= P\LPr{ \sum_{i=1}^{m} \mvec{y}_{\mvec{u}}(i)\mvec{z}(i) \ge \frac{m\sigma_a^2 \tau}{2} } \nonumber\\
&\le \exp\LKr{ - \min \LPr{ \frac{3m\sigma_a^4\tau^2}{16\times 8 K^2\sigma_n^2}, \frac{\sqrt{3} m\sigma_a^2\tau}{4K\sqrt{8} \sigma_n} } }\nonumber\\
&\le \exp\LKr{ - \min \LPr{ \frac{m\sigma_a^4\tau^2}{16\times 3 K^2\sigma_n^2}, \frac{ m\sigma_a^2\tau}{4\sqrt{3}K \sigma_n} } },
\end{align*}
where the last line follows because ${3\over 8}>{1\over 3}$.
Therefore, by the union bound, since $|\C'|\leq 2^{2r}$,
\begin{align*}
P\LPr{\E_3^c} \le2^{2r} \exp\LKr{ -\frac{ m\sigma_a^2\tau}{4\sqrt{3}K \sigma_n}  \min \LPr{ \frac{ \sigma_a^2\tau}{4K\sqrt{3} \sigma_n}, 1 } }.
\end{align*}
Choosing 
\[
 \tau={\sigma_n\over \sigma_a^2} \sqrt{96 K^2 r(1+\epsilon) \over m \log {\rm e}  },
 \]
 and given our choice of $m$, it follows that
 \begin{align*}
P\LPr{\E_3^c} \le2^{-2r\epsilon}.
\end{align*}
But, since $ \sqrt{96\over \log {\rm e} }\leq 9$,
 \begin{align*}
 P\LPr{ \exists \;\mvec{u} \in\C'\; {\rm s.t.}\; \frac{2}{m\sigma_a^2}\LPd{\mvec{y}_{\mvec{u}},  \mvec{z}} > {9K\sigma_n\over \sigma_a^2} \sqrt{  r(1+\epsilon) \over m }}\leq P(\E_3^c) \le2^{-2r\epsilon}.
\end{align*}
Finally, combining \eqref{eq:th3-main-bound} with the bounds derived on the three terms on the right hand side of \eqref{eq:th3-main-bound} yields the desired result. 

\section{Conclusions}\label{sec:conclusion}
In this paper, we have studied the problem of designing efficient compression-based compressed sensing recovery algorithms. Specifically, we have proposed C-GD,  an iterative robust-to-noise compression-based compressed sensing algorithm. Given measurements $\mvec{y}=\mmat{A}\mvec{x}+\mvec{z}$ and a compression code with codebook $\mathcal{C}$, at  iteration $k$, C-GD updates its current estimate of $\mvec{x}$, $\mvec{x}^k$, by moving towards the negative of the gradient of the cost function ($f(\mvec{u})=\|\mvec{y}-\mmat{A}\mvec{u}\|^2$) and then projecting  the result onto the set of codewords $\mathcal{C}$. For a given compression code, the projection step can typically   be implemented by applying the compression code's encoder and decoder.   We have proved that, given enough measurements,  with high probability, C-GD has linear convergence rate and is robust to additive Gaussian noise.  In summary, C-GD provides a platform for using  commercial compression code such as JPEG2000 or MPEG4 for compressed sensing of images and videos, respectively. In our simulation results, we have focused on compressed sensing of images and have shown that C-GD combined with the state-of-the-art compression codes yields state-of-the-art compressed sensing performance, both for i.i.d.~Gaussian and partial Fourier measurement matrices. \\


\setcounter{equation}{0}
\renewcommand{\theequation}{\thesection.\arabic{equation}}

\appendices

\section{ Finding the best piecewise polynomial approximation}\label{app:dynamicprogram}
Consider the following problem: given $\mvec{x}\in\mathds{R}^n$, find $\tilde{\mvec{x}}\in\P$,  $\P$ defined in \eqref{eq:P-def}, such that
\[
\tilde{\mvec{x}}=\arg\min_{\mvec{z}\in\P}\|\mvec{x}-\mvec{z}\|^2_2.
\]
In this section, we briefly describe how $\tilde{\mvec{x}}$ can be found using dynamic programing. Note that, given singularity points $s_1, s_2, \ldots, s_Q$, one can easily find  the best polynomial fit in each piece. Hence, the challenge is to find the optimal singularity points. Each singularity point  $s_i$ is a point in the set $\{{1\over n},\ldots, {n-1\over n}\}$. Given $i_1,i_2\in\{0,\ldots,n\}$,  $i_1\leq i_2$, let $e(i_1,i_2)$ denote the minimum error achievable in approximating $(x_{i_1},\ldots,x_{i_2})$ as samples of $\sum_{j=0}^N a_j y^j$ at ${i_1\over n},\ldots,{i_2\over n}$, where $\sum_{j=0}^n a_j\leq 1$ and $a_j\in(0,1)$, $j=0,\ldots,N$.  That is,
\begin{align}
e(i_1,i_2)=\min_{a_0,\ldots,a_N: a_j\in(0,1), \sum_{j=0}^Na_j\leq 1}\sum_{k=i_1}^{i_2}\left(x_k-\sum_{j=0}^N a_j \left({k\over n }\right)^j\right)^2.\label{eq:e-i1-i2}
\end{align}
Using this definition, given singularity points $s_0=0,s_1,\ldots,s_Q,s_{Q+1}=1\in\{0,{1\over n},\ldots,1\}$, the minimum achievable error in approximating $\mvec{x}$ by signals in $\P$ whose singularities happen at $s_1,\ldots,s_Q$ can be written as
\[
\sum_{i=1}^{Q+1} e(ns_{i-1},ns_i).
\]
This representation suggests that the minimizer $\tilde{\mvec{x}}$ can be found using the Viterbi algorithm. In summary, the Viterbi algorithm will operate on a Trellis diagram with $Q$ full stages corresponding to the possible $Q$ singularities and two single-state stages corresponding to the start and the end of the interval. Each intermediate stage has $n-1$ states, which correspond to the possible $n-1$ singularity points. State $s$ at stage $t\in{1,\ldots,Q}$ is connected to state $s'$ at state $t+1$, if $s\leq s'$. The weight of this edge is set as $e(ns,ns')$, defined in \eqref{eq:e-i1-i2}. Otherwise, if $s'>s$, there is no edge between the two states. 
Let $E_i(s)$ denote the  minimum cost associated with state $s$ at stage $i$. Also, let $E_0(s_0)=0$. The goal is to find the path on the Trellis diagram that achieves $E_{Q+1}(s_{Q+1})=E_{Q+1}(1)$. It is straightforward to show that, for $t=1,\ldots,Q$,  
\[
E_{t+1}(s)=\min_{s'} (E_t(s')+e(s',s)),
\] 
where the minimum is taken over all states $s'$   which are connected to  $s$, i.e., $s'<s$. This breakdown of the cost function describes  the essence of how the Viterbi algorithm operates. At  stage $t$, among its incoming edges, each state $s$ only keeps the edge that achieve $E_{t+1}(s)$. At the end, backtracking from the final state $s_{Q+1}=1$ at stage $Q+1$ reveals the optimal singularities.

\bibliographystyle{unsrt}
\bibliography{myrefs}

\end{document}